\documentclass{IEEEtran}

\renewcommand{\baselinestretch}{1.25}
\usepackage[includehead,hmargin={3.6cm, 2.6cm}, vmargin={2.5cm,2.7cm}, headsep=.8cm,footskip=1.2cm]{geometry}
\usepackage{amssymb}
\usepackage{amsmath}
\usepackage{amsthm}
\usepackage{graphicx}
\usepackage{wrapfig}
\usepackage{paralist}
\usepackage{latexsym}
\usepackage{txfonts}
\usepackage{dsfont}
\usepackage{bbold}
\usepackage{cite}
\usepackage{hyperref}
\usepackage{tikz}
\usepackage{pgfplots}
\usepackage{sidecap}
\usepackage{microtype}
\usepackage{pifont}
\usepackage{mathtools}
\usepackage{algorithm,algpseudocode}

\usetikzlibrary{positioning, automata, shapes.arrows, calc, shapes, arrows}

\newcommand{\hide}[1]{}
\newcommand{\pponly}[1]{{}}
\newcommand{\jrronly}[1]{{}}
\newcommand{\rronly}[1]{{#1}}
\newcommand{\inputsynth}[1]{{}}

\renewcommand{\vec}[1]{\boldsymbol{#1}}
\newcommand{\cvec}[2]{\left[\begin{array}{@{}c@{}}{#1}\\{#2}\end{array}\right]}
\newcommand{\qmat}[4]{\left[\begin{array}{@{}cc@{}}{#1}&{#2}\\{#3}&{#4}\end{array}\right]}
\newcommand{\mat}[1]{\boldsymbol{#1}}

\newcommand{\intmat}[1]{\mathbb{#1}}
\newcommand{\res}{\mathit{res}}

\newtheorem{definition}{Definition}
\newtheorem{lemma}{Lemma}
\newtheorem{theorem}{Theorem}
\newtheorem{corollary}{Corollary}
\newtheorem{proposition}{Proposition}

\newtheorem{remark}{Remark}

\title{{Unbounded-Time Analysis of\jrronly{ Continuous} Guarded LTI Systems with Inputs by\jrronly{ Counterexample Guided} Abstract Acceleration }}
  
\author{Dario Cattaruzza \and Alessandro Abate \and Peter Schrammel \and Daniel Kroening}

\begin{document}
\maketitle

\begin{abstract}
Reachability analysis of continuous and discrete time systems is a hard problem that has seen much progress in the last decades.
In many cases the problem has been reduced to bisimulations with a number of limitations in the nature of the dynamics, soundness, or time horizon. In this article we focus on sound safety verification of Unbounded-Time Linear Time-Invariant (LTI) systems with inputs \jrronly{for both continuous and discrete time} using reachability analysis. We achieve this by using \jrronly{Counterexample-Guided} Abstract Acceleration, which over-approximates the reach tube of a system over unbounded time by using abstraction \jrronly{and finding concrete counterexamples for abstraction refinement based on the safety specification}. The technique is applied to a number of \jrronly{discrete and continuous} models and the results show good performance when compared to state-of-the-art tools.
\end{abstract}

\section{Introduction} \label{sec:intro}

Linear loops are an ubiquitous programming pattern. Linear loops iterate
over continuous variables, which are updated with a linear transformation. 
Linear loops may be guarded, i.e., terminate if a given linear condition
holds.  Inputs from the environment can be modelled by means of
non-deterministic choices within the loop.  These features make linear loops
expressive enough to capture the dynamics of many hybrid dynamical models. 
The usage of such models in safety-critical embedded systems makes linear
loops a fundamental target for formal methods.

Many high-level requirements for embedded control systems can be modelled as
safety properties, \emph{i.e.} deciding reachability of certain \emph{bad
states}, in which the system exhibits unsafe behaviour.  Bad states may, in
linear loops, be encompassed by guard assertions.  

Reachability in linear programs, however, is a formidable challenge for automatic analysers:
the problem is undecidable despite the restriction to linear transformations
(i.e., linear dynamics) and linear guards.  

The goal of this article is to push the frontiers of unbounded-time
reachability analysis: we aim at devising a method that is able to
reason soundly about unbounded trajectories.  We present a new
approach for performing \emph{abstract acceleration}.  Abstract
acceleration~\cite{GH06,JSS14,GS14} approximates the effect of an
arbitrary number of loop iterations (up to infinity) with a single,
non-iterative transfer function that is applied to the entry state of
the loop (i.e., to the set of initial conditions of the linear
dynamics).  This article extends the work in~\cite{JSS14} to systems
with non-deterministic inputs \jrronly{and to continuous time models,}
elaborating the details omitted in~\cite{Sch15}.

The key contributions of this article are:

\begin{enumerate}
\item We present a new technique to include inputs (non-determinism) in the
abstract acceleration of general linear loops.
\item We introduce the use of support functions in complex spaces, in order
to increase the precision of previous abstract acceleration methods.
\jrronly{
\item We develop a counterexample-guided refinement for Abstract Acceleration for safety verification, 
maximising speed when precision is not necessary, thus allowing for optimal analysis within a safe region.
\item We extend Abstract Acceleration to the continuous time case.
}
\end{enumerate}

\section{Preliminaries} 

\subsection{Linear Loops with Inputs}\label{sec:linear_loops}

Simple linear loops are programs expressed in the form:
\begin{displaymath}
\mathit{while} (\mat{G}\vec{x} \leq \vec{h})\ \ \vec{x} := \mat{A}\vec{x}+\mat{B}\vec{u}, 
\end{displaymath}
where $\vec{x} \in \mathbb{R}^p$ is a valuation on the state variables, 
$\psi := \mat{G}\vec{x} \leq \vec{h}$ is a linear constraint on the states (with
$\mat{G} \in \mathbb{R}^{r \times p}$ and $\vec{h} \in \mathbb{R}^r$), 
$\vec{u} \in \mathbb{R}^q$ is a non-deterministic input, 
and $\mat{A} \in R^{p \times p}$ and $\mat{B} \in \mathbb{R}^{p \times q}$ are linear
transformations characterising the dynamics of the system. 
In particular, the special instance where $\psi = \top$ (i.e., ``while
true'') represents a time-unbounded loop with no guards, for which the
discovery of a suitable invariant (when existing) is paramount.
As evident at a semantical level, this syntax can be
interpreted as the dynamics of a discrete-time LTI model with inputs, under
the presence of a guard set which, for ease of notation, we denote as $G =
\{\vec{x}\mid \mat{G}\vec{x}\leq\vec{h}\}$.
In the remaining of this work we will also use the notation $\mat{M}_{i,*}$
to represent the rows of a matrix and $\mat{M}_{*,j}$ its columns.

\subsection{Model Semantics}\label{sec:model_semantics}

The traces of the model starting from an initial set $X_0\subseteq \mathbb{R}^p$, 
with inputs restricted to $U \subseteq \mathbb{R}^q$, are sequences 
$ \vec{x}_0 \xrightarrow{\vec{u}_0} \vec{x}_1 \xrightarrow{\vec{u}_1} \vec{x}_2 \xrightarrow{\vec{u}_2} \ldots $, 
where $ \vec{x}_0 \in X_0$ and $\forall k\geq 0, \vec{x}_{k+1} = \tau(\vec{x}_k,\vec{u}_k) $, 
where 
\begin{equation}\label{equ:reachtraj}
\tau(\vec{x}_k,\vec{u}_k) = 
\left \{ \mat{A}\vec{x}_k + \mat{B}\vec{u}_k
\mid \mat{G}\vec{x}_k \leq \vec{h} \wedge \vec{u}_k \in U \right \} \;
\end{equation}

We extend the notation above to convex sets of states and inputs ($X$ and $U$), 
and denote the set of states reached from $X$ by $\tau$ in one step:
\begin{equation}
\tau(X,U)=\left\{\tau(\vec{x},\vec{u}) \mid \vec{x} \in
X, \vec{u} \in U \right\}
\end{equation}
We furthermore denote the set of states reached from
$X_0$ via $\tau$ in $n$ steps (\emph{$n$-reach set}), for $n\geq 0$:
\begin{align}
\tau^0(X_0,U)&=X_0\nonumber\\
\tau^n(X_0,U) &= \tau(\tau^{n-1}(X_0,U) \cap G,U)
\label{equ:reachset}
\end{align}
Since the transformations $\mat{A}$ and $\mat{B}$ are linear, and
vector sums preserve convexity, the sets $X_n = \tau^n(X_0,U)$ are also
convex.

We define the \emph{$n$-reach tube} 
\begin{equation}\label{equ:reachtube}
\hat{X}_n=\hat{\tau}^n(X_0,U)=\bigcup_{k\in[0,n]} \tau^k(X_0,U)
\end{equation}
as the union of the reachable sets over $n$ iterations.
Moreover, $\hat{X} 
=\bigcup_{n\geq 0} \tau^n(X_0,U)$ 
extends the previous notion over an 
unbounded time horizon.

\subsection{Support Functions} \label{sec:support}

\subsubsection{Support Function Definition} \label{sec:support_def}

A support function is a convex function on $\mathbb{R}^p$ which describes the distance of a supporting
hyperplane for a given geometrical set in $\mathbb{R}^p$. 

Support functions may be used to describe a set by defining the distance of
its convex hull with respect to the origin, given a number of directions. 
More specifically, the distance from the origin to the hyperplane that is orthogonal to the given
direction and that touches its convex hull at its farthest. Finitely sampled 
support functions are template polyhedra in which the directions are not fixed, 
which helps avoiding wrapping effects \cite{GLM06}.
The larger the number of directions provided, the more precisely represented the set will be. 

In more detail, given a direction $\vec{v} \in \mathbb{R}^p$,
the support function of a non-empty set $X \subseteq
\mathbb{R}^p$ in the direction of $\vec{v}$ is defined as 
\begin{displaymath}
\rho_X : \mathbb{R}^p \to \mathbb{R}, \quad \rho_X(\vec{v}) = \sup\{ \vec{x} \cdot \vec{v} : \vec{x} \in X\} \;.
\end{displaymath}
where $\vec{x} \cdot \vec{v}=\sum_{i=0}^p x_i v_i$ is the dot product of the two vectors.

Support functions do not exclusively apply to convex polyhedra, 
but in fact to any set $X \subseteq \mathbb{R}^p$ represented by a general assertion $\theta(X)$. 
We will restrict ourselves to the use of 
convex polyhedra, in which case the support function definition translates to solving the linear program
\begin{equation}
\rho_X(\vec{v}) = \max\{ \vec{x} \cdot \vec{v} \mid \mat{C}\vec{x} \leq \vec{d}\} \;. 
\end{equation}
$\ $\\
\subsubsection{Support Functions Properties}\label{sec:support_properties}

Several properties of support functions allow us to reduce operational
complexity.  The most significant
are~\cite{DBLP:dblp_journals/cviu/GhoshK98}:
\begin{displaymath}
\begin{array}{l}
\rho_{kX}(\vec{v}) =\rho_X(k \vec{v}) = k \rho_X(\vec{v}) : k \geq 0 \\
\rho_{AX}(\vec{v})=\rho_{X}(A^T\vec{v}) : A \in \mathbb{R}^{p \times p}\\
\rho_{X_1 \oplus X_2}(\vec{v})=\rho_{X_1}(\vec{v})+\rho_{X_2}(\vec{v}) \\
\rho_X(\vec{v}_1 + \vec{v}_2) \leq \rho_X({\vec{v}_1}) + \rho_X({\vec{v}_2})\\
\rho_{conv(X_1 \cup X_2)}(\vec{v}) = \max\{\rho_{X_1}(\vec{v}),\rho_{X_2}(\vec{v})\}\;\;\;\;\;\; \\
\rho_{X_1 \cap X_2}(\vec{v}) \leq \min\{\rho_{X_1}(\vec{v}),\rho_{X_2}(\vec{v})\}\\ 
\end{array}
\end{displaymath}
As can be seen by their structure, some of these properties reduce
complexity to lower-order polynomial 
%
or even to constant time, by turning matrix-matrix multiplications ($\mathcal{O}(p^3)$) into matrix-vector ($\mathcal{O}(p^2)$), 
or into scalar multiplications.

\subsubsection{Support Functions in Complex Spaces}\label{sec:support_complex}

The literature does not state, as far as we found any description of the use of support functions in complex spaces. 
Since this is relevant to using our technique, we extend the definition of support functions to encompass their operation
on complex spaces.\\
\\
A support function in a complex vector field is a transformation:
\begin{displaymath}
\rho_X(\vec{v}) : \mathbb{C}^p \to \mathbb{R} = \sup\{|\vec{x} \cdot \vec{v}| \mid \vec{x} \in X \subseteq \mathbb{C}^p, \vec{v} \in \mathbb{C}^p\}. 
\end{displaymath}
The dot product used here is the Euclidean Internal Product of the vectors,
which is commonly defined in the complex space as:
\begin{displaymath}
\vec{a} \cdot \vec{b} = \sum_{i=0}^p a_i\overline{b}_i\ , \ \ \ \vec{a}, \vec{b} \in \mathbb{C}^p
\end{displaymath}
We are interested in the norm of the complex value, which is a 1-norm given our definition of dot product:
\begin{displaymath}
|\vec{a} \cdot \vec{b}|=|re(\vec{a} \cdot \vec{b}) |+ |im(\vec{a} \cdot \vec{b}) |
\end{displaymath}

 Returning to our support function properties, we now have:
\begin{displaymath}
\rho_X(r e^{i\theta} \vec{v}) = r \rho_X(e^{i\theta}\vec{v}),  
\end{displaymath}
which is consistent with the real case when $\theta=0$.  
The reason why $e^{i\theta}$ cannot be extracted out is because it is
a rotation, and therefore follows the same rules as a matrix
multiplication,  
\begin{displaymath}
\rho_X(e^{i\theta}\vec{v}) \triangleq \rho_X
\left (
\left [ \begin{array}{cc}
\cos \theta & \sin \theta\\
-\sin \theta & \cos \theta
\end{array}
\right ]
\vec{v}
\right ). 
\end{displaymath}
Since matrices using pseudo-eigenvalues are real, all other properties remain the same. 
An important note is that when using pseudo-eigenvalues, 
conjugate eigenvector pairs must be also converted into two separate real eigenvectors, 
corresponding to the real and the imaginary parts of the pair. 


\section{The Polyhedral Abstract Domain} \label{sec:polyhedra}

\subsection{Convex Polyhedra} \label{sec:convex_polyhedra}

A polyhedron is a topological element in $\mathbb{R}^p$ with flat polygonal (2-dimensional) faces. Each 
face corresponds to a hyperplane that creates a halfspace, and the intersections of these hyperplanes are the edges of the polyhedron.
A polyhedron is said to be convex if its surface does not intersect itself and a line segment joining any two 
points of its surface is contained in the interior of the polyhedron. Convex polyhedra are better suited than
general polyhedra as an abstract domain, mainly because they have a simpler representation and operations
over convex polyhedra are in general easier than for general polyhedra.
There are a number of properties of convex polyhedra that make them ideal for abstract interpretation of
continuous spaces, including their ability to reduce an uncountable set of real points into a countable set
of faces, edges and vertices.
Convex polyhedra retain their convexity across linear transformations, and are functional across a number
of operations because they have a dual representation~\cite{fukuda1996double}. The mechanism to switch between these two
representations is given in section \ref{sec:vertex}

\subsubsection{Vertex Representation} \label{sec:vertex_representation}
Since every edge in the polyhedron corresponds to a line between two vertices and every face corresponds to
the area enclosed by a set of co-planar edges, a full description of the polyhedron is obtain by simply listing its 
vertices.
Since linear operations retain the topological properties of the polyhedron, performing these operations on the
vertices is sufficient to obtain a complete description of the transformed polyhedron (defined by the transformed
vertices).
Formally, a polyhedron is a set $V \in \mathbb{R}^p$ such that $\vec{v} \in V$ is a vertex of the polyhedron. 

\subsubsection{Inequality Representation} \label{sec:ine_representation}

The dual of the Vertex representation is the face representation. Each face corresponds to a bounding hyperplane
of the polyhedron (with the edges being the intersection of two hyperplanes and the vertices the intersection of 3 or
more), and described mathematically as a function of the vector normal to the hyperplane.
If we examine this description closely, we can see that it corresponds to the support function of the vector normal
to the hyperplane.
Given this description we formalise the following:
A convex polyhedron is a topological region in $\mathbb{R}^p$ described by the set 
\begin{align*}
X=\{\vec{x} \in \mathbb{R}^{p} \mid \mat{C}\vec{x} \leq \vec{d}, \mat{C} \in \mathbb{R}^{m \times p}, \vec{d} \in \mathbb{R}^{m}\}
\end{align*}
where the rows $\mat{C}_{i,*}$ for $i \in [1,m]$ correspond to the transposed vectors normal to the faces of the polyhedron and
$\vec{d}_i$ for $i \in [1,m]$ their support functions.
For simplicity of expression, we will extend the use of the support function operator as follows:
\begin{align*}
\rho'_X &: \mathbb{R}^{m \times p} \to \mathbb{R}^m\\
\rho'_X(\mat{M}^T) &= \left [ \begin{array}{c} \rho_X(\vec{\mat{M}_{1,*}^T)}\\ \rho_X(\vec{\mat{M}_{2,*}^T)}\\ \vdots \\ \rho_X(\vec{\mat{M}_{m,*}^T})\end{array}\right]
\end{align*}

\subsection{Operations on Convex Polyhedra} \label{sec:convex_polyhedra_ops}

Several operations of interest can be performed on convex polyhedra
\subsubsection{Translation} \label{sec:convex_translation}
Given a vertex representation $V$ and a translation vector $\vec{t}$, the transformed polyhedron is
\begin{displaymath}
V'=\{ \vec{v}'= \vec{v}+\vec{t} \mid \vec{v} \in V \}
\end{displaymath}
Given an inequality representation $X$ and a translation vector $\vec{t}$, the transformed polyhedron corresponds to
\begin{displaymath}
X'= \left\{ \vec{x} \mid \mat{C}\vec{x} \leq \vec{d}+\mat{C}\vec{t}\right\}
\end{displaymath}

\subsubsection{Linear Transformation} \label{sec:convex_transform}
Given a vertex representation $V$ and a linear transformation $\mat{L}$, the transformed polyhedron is
\begin{displaymath}
\mat{V}'=\mat{L}V
\end{displaymath}
Given an inequality representation $X$ and a linear transformation $\mat{L}$,
the transformed polyhedron corresponds to
\begin{displaymath}
X'=\left\{\vec{x} \mid \mat{C}(\mat{L}^+)^T\vec{x} \leq \rho'_X(\mat{L}^+\mat{C}^T)\right\}
\end{displaymath}
where $\mat{L}^+$ represents the pseudo-inverse of $\mat{L}$. In the case when the inverse $\mat{L}^{-1}$ exists, then
\begin{displaymath}
X'=\left\{\vec{x} \mid \mat{C}(\mat{L}^{-1})^T\vec{x} \leq \vec{d})\right\}
\end{displaymath}
From this we can conclude that linear transformations are better handled by a vertex representation, except when the inverse of the transformation exists and is know a-priori. This works makes use of this last case to avoid continuous swapping in representations.

\subsubsection{Set Sums} \label{sec:convex_sum}
The addition of two polyhedra is a slightly more complex matter. The resulting set is one such that for all possible combinations of points inside both original polyhedra, the sum is contained in the result. This operation is commonly known as the Minkowski sum, namely
\begin{align*}
A \oplus B = \{a + b\mid a\in A, b\in B\}
\end{align*}
Given two vertex representations $V_1$ and $V_2$ the resulting polyhedron
\begin{align*}
V=conv(V_1 \oplus V_2)
\end{align*}
where $conv(\cdot)$ is the convex hull of the set of vertices contained in the Minkowski sum.\\
Let
\begin{align*}
X_1&=\{ \vec{x} \mid \mat{C}_1\vec{x} \leq \vec{d}_1 \}\\
X_2&=\{ \vec{x} \mid \mat{C}_2\vec{x} \leq \vec{d}_2 \}
\end{align*}
be two sets, then
\begin{align*}
X &=X_1 \oplus X_2 = \{ \vec{x} \mid \mat{C}\vec{x} \leq \vec{d} \} ,
\end{align*}
where
\begin{equation*}
\mat{C}=\left[\begin{array}{c}\mat{C}_1\\ \mat{C}_2\end{array}\right],\ \vec{d} = \left[\begin{array}{c}\vec{d}_1+\rho'_{X_2}(\mat{C}_1^T)\\ \vec{d}_2+\rho'_{X_1}(\mat{C}_2^T)\end{array}\right].
\end{equation*}
Because these sets correspond to systems of inequalities, they may be reduced removing redundant constraints. Note that if $\mat{C}_1=\mat{C}_2$ then
\begin{align*}
X &=X_1 \oplus X_2 = \{ \vec{x} \mid \mat{C}_1\vec{x} \leq \vec{d}_1+\vec{d}_2 \} ,
\end{align*}

\subsubsection{Set Hadamard Product} \label{sec:convex_prod}
\begin{lemma}
Given two vertex representations $V_1$ and $V_2$ the resulting polyhedron
\begin{align*}
V=V_1 \circ V_2 =conv(\{ \vec{v}=\vec{v}_1 \circ \vec{v}_2 \mid \vec{v}_1 \in V_1, \vec{v}_2 \in V_2)
\end{align*}
where $\circ$ represents the Hadamard (coefficient-wise) product of the vectors, contains all
possible combinations of products between elements of each set.
\end{lemma}
\begin{proof}
Given a convex set $X$, we have:
$$X'=\{\vec{x}_{ij} \mid \vec{x}_i,\vec{x}_j \in X, \vec{x}_{ij}=t\vec{x}_i + (1-t)\vec{x}_j, t\in[0,1] \} \subseteq X$$
Given $\vec{x}_i \in X$, $\vec{y}_j \in Y$, $\vec{z}_{i,j}=\vec{x}_i \circ \vec{y}_j \in Z$
\begin{align*}
\vec{x}_{ij} \in X', \vec{y}_k \in Y, \vec{z}_{i,k}, \vec{z}_{j,k} \in Z &\Rightarrow \vec{z}_{ij,k} \in Z\\
\vec{x}_{ij} \in X, \vec{y}_m, \vec{y}_n \in Y, \vec{z}_{ij,m}, \vec{z}_{ij,n} \in Z &\Rightarrow \vec{z}_{ij,mn} \in Z
\end{align*}
This equation proves that given $\vec{v}_{11}, \vec{v}_{12} \in V_1$, $\vec{v}_{21}, \vec{v}_{22} \in V_2$ and $u,t \in [0, 1]$,
$$\left(t\vec{v}_{11}+(1-t)\vec{v}_{12}\right)\left(u\vec{v}_{21}+(1-u)\vec{v}_{22}\right) \in V$$
\end{proof}

\subsubsection{Vertex Enumeration} \label{sec:vertex}

The vertex enumeration problem corresponds to the algorithm required to obtain a list of all vertices of a
polyhedron given an inequality description of its bounding hyperplanes. Given the duality of the problem, it is
also possible to find the bounding hyperplanes given a vertex description if the chosen algorithm exploits this duality.
In this case the description of V is given in the forms of a matrix inequality $\mat{V}\vec{x} \leq [\begin{array}{cccc}1&1&\cdots&1\end{array}]^T$ with $\mat{V}=[\begin{array}{ccc}\vec{v}_1&\cdots&\vec{v}_m\end{array}]^T, \vec{v}_i \in V$. Similarly, $\mat{A}$ can be described as a set containing each of its rows. 
At the time of writing, there are two algorithms that efficiently solve the  vertex enumeration problem. lrs is
a reverse search algorithm, while cdd follows the double description method. 
In this work we use the cdd algorithm for convenience in implementation (the original cdd was developed 
for floats, whereas lrs uses rationals). The techniques presented here can be applied to either.

Let 
\begin{align*}
\mathcal{C} = \{ \vec{x} \mid \mat{A}\vec{x} \geq 0, \mat{A} \in \mathbb{R}^{n \times p}, \vec{x} \in \mathbb{R}^p\}
\end{align*}
be the polyhedral cone represented by $\mat{A}$. 
The pair $(\mat{A},V)$ is said to be a double description pair if
$$ \mathcal{C}=\{\vec{\lambda}^T V \mid V \in \mathbb{R}^p, \vec{\lambda} \in \mathbb{R}_{\geq 0}^{|V|}\}$$
$V$ is called the generator of $X$.
Each element in $V$ lies in the cone of $X$, and its minimal form (smallest $m$) has a one-to-one
correspondence with the extreme rays of X if the cone is pointed (i.e., it has a vertex at the origin).
This last can be ensured by translating a polyhedral description so that it includes the origin, and then 
translating the vertices back once they have been discovered (see section \ref{sec:convex_polyhedra_ops}).

We will also point out that
\begin{align*}
&\left\{\vec{x} \mid \mat{A}\vec{x} \leq \vec{b}\right\}=\left\{ \vec{x}' \mid [\begin{array}{cc}-\mat{A}&\vec{b}\end{array}]\,\vec{x}' \geq 0\right\}\\
&\text{where }\vec{x} \in \mathbb{R}^p\text{ and }\vec{x}'= \left[\begin{array}{c}\vec{x}\\1\end{array}\right]\in \mathbb{R}^{p+1}.
\end{align*}

The vertex enumeration algorithm starts by finding a base $\mathcal{C}_K$ which contains a number of
vertices of the polyhedron. This can be done by pivoting over a number of different rows in $\mat{A}$ and selecting the 
feasible visited points, which are known to be vertices of the polyhedron (pivoting $p$ times will ensure at least one
vertex is visited if the polyhedron is non-empty). $\mathcal{C}_K$ is represented by  $\mat{A}_K$ which contains the rows 
used for the pivots.
The base $\mathcal{C}_K$ is then iteratively expanded to $\mathcal{C}_{K+i}$ by exploring the $i^{th}$ row of
$\mat{A}$ until $\mathcal{C}_K= \mathcal{C}$. The corresponding pairs $(\mat{A}_{K+i},V_{K+i})$ are
constructed using the information from $(\mat{A}_K,V_K)$ as follows:\\
Let $\mat{A}_K \in \mathbb{R}^{n_K \times p}$, $\mat{A}_{i,*} \in \mathbb{R}^{1 \times p}$, $V_K \in \mathbb{R}^p$,
\begin{align} 
H_i^+&=\{\vec{x} \mid \mat{A}_{i,*}\vec{x} >0 \},\\
H_i^-&=\{\vec{x} \mid \mat{A}_{i,*}\vec{x} <0 \},\\ 
H_i^0&=\{\vec{x} \mid \mat{A}_{i,*}\vec{x} =0 \}
\label{eq:dd_hyperplanes}
\end{align}
be the spaces outside inside and on the $i^{th}$ hyperplane and 
\begin{align} 
V_K^+&=\{\vec{v}_j \in H_i^+\},\\
V_K^-&=\{\vec{v}_j \in H_i^-\},\\ 
V_K^0&=\{\vec{v}_j \in H_i^0\}
\label{eq:dd_partial_vertices}
\end{align}
the existing vertices lying on each of these spaces.\\
Then 
\begin{align}
\label{eq:vertex_enum}
V_{K+i}&=V_K^+ \cup V_K^- \cup V_K^0 \cup  V_K^i\\
V_K^i&=\left\{(\mat{A}_{i,*}\vec{v}^+)\vec{v}^--(\mat{A}_{i,*}\vec{v}^-)\vec{v}^+ \mid \vec{v}^- \in V^-, \vec{v}^+ \in V^+ \right\}\nonumber\\
\end{align}
For the proof see \cite{fukuda1996double}.

\section{Abstract Matrices in abstract acceleration}\label{sec:absacc}

\subsection{Acceleration Techniques}\label{sec:accel}

Acceleration of a transition system is a method that seeks to precisely describe the transition relations over
a number of steps using a concise description of the overall transition between the first and final step. Namely,
it looks for a direct formula to calculate the postimage of a loop from the initial states of the loop.
\jrronly{It may be applied to both continuous and discrete systems using different paradigms.}
Formally, given the dynamics in equation \eqref{equ:reachtraj} an acceleration formula aims at computing
the reachset \eqref{equ:reachset} using a function $f$ such that $f(\cdot)=\tau^n(\cdot)$. In the case
of systems without inputs, this equation is $\vec{x}_n=\mat{A}^n\vec{x}_0$. We will use this property
and others derived from it to calculate our abstract matrices.\\

\subsection{Overview of the Algorithm}\label{sec:absaccel_overview}

The basic steps required to evaluate a reach tube using abstract acceleration
can be seen in figure~\ref{fig:procedure}. 
\begin{enumerate}
\item The process starts by doing
eigendecomposition of the dynamics ($\mat{A}$) in order to transform the
problem into a simpler one. 
\item A variety of off-the-shelf tools may be used, but
since larger problems require numerical algorithms for scalability, a second
step involves upper-bounding the error in order to obtain sound results. 
In such cases, all subsequent steps must be performed using interval arithmetic.
\item The inverse of the generalised eigenvectors must be calculated soundly.
\item The problem gets transformed into canonical form by multiplying both sides of the equation by $\mat{S}^{-1}$:
\begin{align*}
X_k'&=\intmat{J}\left (X_{k-1}' \cap G' \right)+U' \text{ where}\\
X'&=\intmat{S}^{-1}X, U'=\intmat{S}^{-1}\mat{B}U \text{ and } G'=\{\vec{x} \mid \mat{G}\intmat{S}\vec{x} \leq \vec{h} \}
\end{align*}
\item We calculate the number of iterations as explained in section~\ref{sec:guards}. If there are no guards, we use $n=\infty$. It is worth noting that this number need not be exact: if we overapproximate the number of iterations, the resulting reachtube will overapproximate the desired one.
\item we overapproximate the dynamics of the variable inputs (for parametric or no inputs this step will be ignored) using the techniques described in section~\ref{sec:var_inpur_accel}
\item we calculate the abstract dynamics using the techniques described in section~\ref{sec:support_aa}
\item we evaluate the vertices of the combined input-initial eigenspace to be used as source for the reachtube calculation
\item we use a sound simplex algorithm to evaluate the convex set product of the abstract dynamics (used as the tableau) and the initial set (whose vertices are used as the obejctive functions alongside a set of template directions for the result).
\item since we have calculated our result in the eigenspace, we transform the reachtube back into the normal space by multiplying by $\intmat{S}$ .
\end{enumerate}
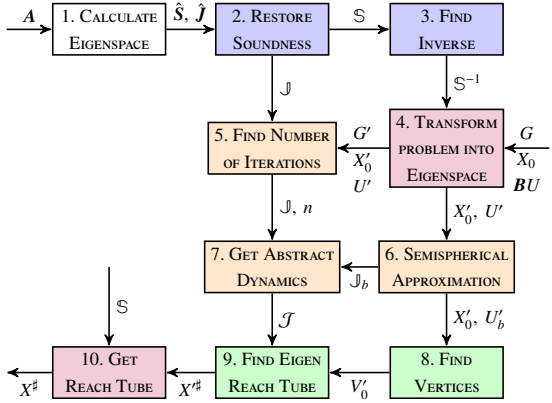
\begin{figure}
\centering
{\scriptsize
\begin{tikzpicture}[scale=0.3,->,>=stealth',shorten >=.2pt,auto, semithick, ampersand replacement=\&,]
  \matrix[nodes={draw, fill=none, shape=rectangle, minimum height=.2cm, minimum width=1.5cm, align=center},row sep=.7cm, column sep=.5cm] {
   \& \node[fill=white!20,align=center] (eigen) {{\sc 1. Calculate}\\{\sc Eigenspace}};
   \& \node[fill=blue!20,align=center] (sound_eigen) {{\sc 2. Restore}\\{\sc Soundness}};
   \& \node[fill=blue!20,align=center] (inverse) {{\sc 3. Find}\\{\sc Inverse}};\\
   \& \coordinate (aux1);
   \& \node[fill=orange!20,align=center] (iters) {{\sc 5. Find Number}\\{\sc of  Iterations}};
   \& \node[fill=purple!20,align=center] (init) {{\sc 4. Transform}\\{\sc problem into}\\{\sc Eigenspace}};\\
   \& \coordinate (aux);
   \& \node[fill=orange!20,align=center] (aaj) {{\sc 7. Get Abstract}\\{\sc Dynamics}};
   \& \node[fill=orange!20,align=center] (sphere) {{\sc 6. Semispherical}\\{\sc Approximation}};\\
   \& \node[fill=purple!20,align=center] (reach) {{\sc 10. Get}\\{\sc Reach Tube}};
   \& \node[fill=green!20,align=center] (eigenreach) {{\sc 9. Find Eigen}\\{\sc Reach Tube}};
   \& \node[fill=green!20,align=center] (vert) {{\sc 8. Find}\\{\sc Vertices}};\\
  };
  \path
     ([xshift=-2cm]eigen.west) edge node[align=center] {$\mat{A}$} (eigen.west)
    (eigen.east) edge node[align=center] {$\hat{\mat{S}}$, $\hat{\mat{J}}$} (sound_eigen.west)
    (sound_eigen.east) edge node[align=center] {$\intmat{S}$} (inverse.west);
  \path
    (sound_eigen.south) edge node[align=center] {$\intmat{J}$} (iters.north)
    (inverse.south) edge node[align=center] {$\intmat{S}^{-1}$} (init.north);
  \path
    (init.west) edge node[align=center, yshift=.38cm] {{$G'$}\\{$X_0'$}\\{$U'$}} (iters.east);
  \path
    (iters.south) edge node[align=center] {$\intmat{J}$, $n$} (aaj.north)
    (init.south) edge node[align=center] {{$X_0'$, $U'$}} (sphere.north);
  \path 
    (sphere.south) edge node[align=center] {{$X_0'$, $U_b'$}} (vert.north)
    (sphere.west) edge node[align=center] {{$\intmat{J}_b$}} (aaj.east)
    (aaj.south) edge node[align=center] {$\mathcal{J}$} (eigenreach.north);
   \path
    ([xshift=2cm]init.east) edge node[align=center, yshift=.38cm] {{$G$}\\{$X_0$}\\{$\mat{B}U$}} (init.east)
    (vert.west) edge node[align=center] {{$V_0'$}} (eigenreach.east)
    (eigenreach.west) edge node[align=center] {{$X'^\sharp$}} (reach.east)
    (aux) edge node[align=center] {$\intmat{S}$}  (reach.north)
    (reach.west) edge node[align=center] {$X^\sharp$} ([xshift=-2cm]reach.west);
\end{tikzpicture}
}
\caption{Block diagram describing the different steps used to calculate the abstract reach tube of a system.}
\label{fig:procedure}
\end{figure}
\subsection{Computation of Abstract Matrices}\label{sec:absmat_real}

We define the abstract matrix $\mathcal{A}^n$ as an over-approximation
of the union of the powers of the matrix $\mat{A}^k$ such that $\mathcal{A}^n \supseteq
\bigcup_{k\in[0,n]} \mat{A}^k$ and its application to the initial set $X_0$
\begin{equation}
\hat{X}_n^\sharp=\mathcal{A}^nX_0 \supseteq \hat{X}_n
\label{eq:aa_reach_tube}
\end{equation}
Next we explain how to compute such an abstract matrix.
For simplicity, we first describe this computation for matrices $\mat{A}$ with real eigenvalues, 
whereas the extension to the complex case will be addressed in Section~\ref{sec:absmat_complex}. 
Similar to \cite{JSS14}, we first have to compute the Jordan normal form of $\mat{A}$.
Let $\mat{A}=\mat{S}\mat{J}\mat{S}^{-1}$ where $\mat{J}$ is the normal Jordan form of $\mat{A}$, 
and $\mat{S}$ is made up by the corresponding eigenvectors. 
We can then easily compute $\mat{A}^n=\mat{S}\mat{J}^n\mat{S}^{-1}$, where 
\begin{align}
\mat{J}^n=& \left [ \begin{array}{ccc}
\mat{J}_1^n &          & \\
           & \ddots & \\
           &           & \mat{J}_r^n\\
\end{array} \right ]\\ 
\mat{J}_{s \in[1,r]}^n=& \left [ \begin{array}{cccc}
\lambda_s^n  & \binom{n}{1}  \lambda_s^{n-1} & \hdots  & \binom{n}{p_s-1} \lambda_s^{n-p_s+1} \\
& \lambda_s^n  & \binom{n}{1}  \lambda_s^{n-1} & \vdots \\
\vdots & & \ddots & \vdots \\
& &  &\lambda_s^n \\
\end{array} \right ]
\label{jord:pow}
\end{align}
The abstract matrix $\mathcal{A}^n$ is computed as an abstraction over 
a set of vectors $\vec{m}^k \in \mathbb{R}^p, k \in [1\ n]$ of entries of $\mat{J}^k$. 

Let $\vec{I}_s=[\begin{array}{cccc}1&0&\cdots&0\end{array}] \in \mathbb{R}^{p_s}$. 
The vector $\vec{m}^k$ is obtained by the transformation $\varphi^{-1}$:
\begin{equation}
\vec{m}^k=[\begin{array}{ccc}\vec{I}_1 \mat{J}_1^k&\cdots&\vec{I}_r \mat{J}_r^k\end{array}] \in \mathbb{R}^p
\end{equation}
such that 
$\mat{J}^k = \varphi(\vec{m}^k)$. 

If $\mat{J}$ is diagonal \cite{JSS14}, then $\vec{m}^k$ equals the
vector of powers of eigenvalues $[\lambda_1^k,\ldots,\lambda_r^k]$.
An interval abstraction can thus be simply obtained by computing the
intervals
$[\min\{\lambda_s^0,\lambda_s^n\},$ $\max\{\lambda_s^0,\lambda_s^n\}
], s\in[1,r]$.
We observe that the spectrum of the interval matrix $\sigma(\mathcal{A}^n)$ (defined as intuitively) 
is an over-approximation of $\bigcup_{k\in[0,n]} \sigma(\mat{A}^k)$.

In the case of the $s^\mathrm{th}$ Jordan block $\mat{J}_s$ with geometric non-trivial
multiplicity $p_s$ ($\lambda_i = \lambda_{i-1} =\ldots$), observe that
the first row of $\mat{J}_s^n$ contains all (possibly) distinct
entries of $\mat{J}_s^n$.  Hence, in general, the vector section $\vec{m}_s$ is
the concatenation of the (transposed) first row vectors
$\left(\lambda_s^n , \binom{n}{1} \lambda_s^{n-1}, \hdots,
\binom{n}{p_s-1} \lambda_s^{n-p_s+1}\right)^T$ of $\mat{J}_s^n$.

Since the transformation $\varphi$ 
transforms the vector $\vec{m}$ into the shape of \eqref{jord:pow} of $\mat{J}^n$, it is called a \emph{matrix shape} \cite{JSS14}. 
We then define the abstract matrix as
\begin{equation}
\mathcal{A}^n = \{\mat{S}\ \varphi(\vec{m})\ \mat{S}^{-1} \mid  \mat{\Phi}\vec{m} \leq \vec{f} \} \;,  
\label{abs:mat}
\end{equation} 
where the constraint $\mat{\Phi}\vec{m} \leq \vec{f}$ is synthesised
from intervals associated to the individual eigenvalues and to their
combinations.  More precisely, we compute polyhedral relations: 
for any pair of eigenvalues (or binomials) within $\mat{J}$, we find
an over-approximation of the convex hull containing the points 
$$\left \{ \vec{m}^k \mid k \in [1, n] \right \}  \subseteq \left \{\vec{m} \mid \mat{\Phi}\vec{m} {\leq} \vec{f} \right \}$$
%

\subsection{Abstract Matrices in Complex Spaces}\label{sec:absmat_complex}

To deal with complex numbers in eigenvalues and eigenvectors, 
\cite{JSS14} employs the real Jordan form for conjugate eigenvalues 
$\lambda = re^{i\theta}$ and $\lambda^* = re^{-i\theta}$ ($\theta \in [0,\pi]$), so that 
\begin{displaymath}
\qmat{\lambda}{0}{0}{ \lambda^*} \quad \text{ is replaced by } \quad r \qmat{\cos \theta}{\sin \theta}{-\sin \theta}{\cos \theta}. 
\end{displaymath}
Although this equivalence will be of use once we evaluate the progression of the system, 
calculating powers under this notations is often more difficult than handling directly the original matrices with complex values.  

In Section~\ref{sec:absmat_real}, 
in the case of real eigenvalues we have abstracted the entries in the power matrix $\mat{J}_s^n$ by ranges of eigenvalues
$[\min\{\lambda _s^0 \cdots \lambda _s^n\}\ \ \ \max\{\lambda _s^0 \cdots \lambda _s^n\} ]$.  
In the complex case we can do something similar by rewriting
eigenvalues into polar form $\lambda _s = r _s e^{i\theta _s}$ and abstracting
by $[\min\{r_s^0 \cdots r_s^n\}\ \ \ \max\{r_s^0 \cdots r_s^n\} ]e^{i[0\ , \ \min(\theta_s,2\pi)]}$.

\section{General Abstract Acceleration with Inputs} \label{sec:forward_aa}

\subsection{Using Support Functions on Abstract Acceleration}\label{sec:support_aa}

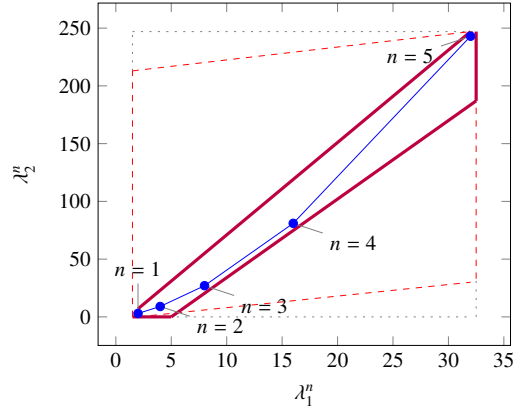
\begin{figure}[]
\centering
\begin{tikzpicture}[scale=0.8]
\begin{axis}[ylabel={$\lambda_2^n$},xlabel={$\lambda_1^n$}]
\addplot[color=blue,mark=*] table{
2 3
4 9
8 27
16 81
32 243
};
\addplot[dash pattern=on 1pt off 3pt , color=gray] coordinates {(5.5,0) (32,0)};
\addplot[dash pattern=on 1pt off 3pt , color=gray] coordinates {(32.5,0) (32.5,33)};
\addplot[dash pattern=on 1pt off 3pt , color=gray] coordinates {(1.5,213) (1.5,247)};
\addplot[dash pattern=on 1pt off 3pt , color=gray] coordinates {(1.5,247) (32.5,247)};
\addplot[dashed, color=red] coordinates {(1.5,0) (1.5,213)};
\addplot[dashed, color=red] coordinates {(1.5,213) (32,247)};
\addplot[dashed, color=red] coordinates {(2,0) (32,30)};
\addplot[dashed, color=red] coordinates {(32.5,33) (32.5,187)};
\addplot[color=purple, line width=1.5pt] coordinates {(2,7) (32,247)};
\addplot[color=purple, line width=1.5pt] coordinates {(5,0) (32.5,187)};
\addplot[color=purple, line width=1.5pt] coordinates {(1.5,0) (5,0)};
\addplot[color=purple, line width=1.5pt] coordinates {(32.5,187) (32.5,247)};
\node at (axis cs:2,3) [pin={+90:$n=1$},inner sep=0pt] {};
\node at (axis cs:4,9) [pin={-10:$n=2$},inner sep=0pt] {};
\node at (axis cs:8,27) [pin={-10:$n=3$},inner sep=0pt] {};
\node at (axis cs:16,81) [pin={-10:$n=4$},inner sep=0pt] {};
\node at (axis cs:32,243) [pin={-170:$n=5$},inner sep=0pt] {};
\end{axis}
\end{tikzpicture}
\caption{Polyhedral faces from an $\mathbb{R}^2$ subspace, where 
  $(\lambda_1^n, \lambda_2^n)$ so that $\lambda_1{=}2, \lambda_2{=}3, 1{\leq}n{\leq}5$. 
  Bold purple lines represent supports found by this article. The dotted grey and dashed
  red polytopes show logahedral approximations (box and octagon)
  used in \cite{JSS14}. Note the scales (sloped dashed lines are
  parallel to the x=y line, and dashed red polytope hides two small faces yielding an octagon).}
\label{aa:supports}
\end{figure}

As an improvement over \cite{JSS14}, the rows in $\mat{\Phi}$ and $\vec{f}$ (see \eqref{abs:mat}) are synthesised by discovering
support functions in these sets. The freedom of directions provided by these support functions results in an improvement over the logahedral abstractions used in previous works (see Figures~\ref{aa:supports} - \ref{aa:mixed_supports}).
The mechanism by which this works follows the convex properties of the exponential progression. 
There are four cases cases to consider \footnote{these explain in detail the procedure alluded in \cite{cattaruzza2015unbounded}}:
\begin{enumerate}
\item Positive Real Eigenvalues\\
The exponential curve is cut along the diagonal between the maximum and minimum eigenvalues to create a support function for the corresponding hyperplane. A third point taken from the curve is used to test the direction of the corresponding template vector. An arbitrary number of additional hyperplanes are selected by picking pairs of adjacent points in the curve and creating the corresponding support function as shown in Figure~\ref{aa:supports}.
\item Complex Conjugate Eigenvalue pairs\\
In the case of complex conjugate pairs, the eigenvalue map corresponds to a logarithmic spiral. In this case, we must first extract the number of iterations required for a full cycle. For convergent eigenvalues, only the first $n$ iterations have an effect on the support functions, while in the divergent case only the last $n$ iterations are considered. Support functions are found for adjacent pairs, checking for the location of the origin point (first point for convergent eigenvalues, last for divergent eigenvalues). If the origin falls outside the support function, we look for an interpolant point that closes the spiral tangent to the origin. This last check is performed as a binary search over the remaining points in the circle (whose supporting planes would exclude the origin) to achieve maximum tightness (see Figure~\ref{aa:complex_supports}).

\begin{figure}[]
\centering
\begin{tikzpicture}[scale=0.8]
\begin{axis}[ylabel={$\lambda_2^n$},xlabel={$\lambda_1^n$}]
\addplot[color=blue,mark=*] table{
0.8 0.4
0.48 0.64
0.128 0.704
-0.1792 0.6144
-0.3891 0.4198
-0.4792 0.1802
-0.4554 -0.0475
-0.3454 -0.2202
-0.1882 -0.3143
-0.0249 -0.3267
0.1108 -0.2713
0.1972 -0.1727
0.2268 -0.0593
0.2052 0.04328
};
\addplot[dash pattern=on 1pt off 3pt , color=gray] coordinates {(-0.48,.72) (0.82,.72)};
\addplot[dash pattern=on 1pt off 3pt , color=gray] coordinates {(-.48,-.33) (0.8,-.33)};
\addplot[dash pattern=on 1pt off 3pt , color=gray] coordinates {(0.82,-.33) (0.82,.72)};
\addplot[dash pattern=on 1pt off 3pt , color=gray] coordinates {(-0.48,-.33) (-0.48,.72)};
\addplot[dash pattern=on 1pt off 3pt, line width=1.5pt , color=blue] coordinates {(0.08,-.33) (0.82,.5)};
\addplot[color=purple, line width=1.5pt] coordinates {(0.82,0.415) (0.82,0.385)};
\addplot[color=purple, line width=1.5pt] coordinates {(0.82,0.415) (.42,.72)};
\addplot[color=purple, line width=1.5pt] coordinates {(-.1,.72) (.42,.72)};
\addplot[color=purple, line width=1.5pt] coordinates {(-.1,.72) (-.48,0.37)};
\addplot[color=purple, line width=1.5pt] coordinates {(-.48,-.05) (-.48,0.37)};
\addplot[color=purple, line width=1.5pt] coordinates {(-.48,-.05) (-.3,-.33)};
\addplot[color=purple, line width=1.5pt] coordinates {(-.3,-.33) (0.08,-0.33)};
\addplot[color=purple, line width=1.5pt] coordinates {(0.82,0.385) (0.08,-0.33)};
\node at (axis cs:0.8,0.4) [pin={+200:$n=1$},inner sep=0pt] {};
\node at (axis cs:0.48,0.64) [pin={-90:$n=2$},inner sep=0pt] {};
\node at (axis cs:.128,.704) [pin={-90:$n=3$},inner sep=0pt] {};
\end{axis}
\end{tikzpicture}
\caption{Polyhedral faces from an $\mathbb{R}^2$ complex conjugate subspace, where 
  $(\lambda_1^n, \lambda_2^n)$ so that $\lambda_1{=}0.8+0.4i, \lambda_2{=}0.8-0.4i, 1{\leq}n{\leq}14$. 
  Bold purple lines represent supports found by this article. The blue dotted line shows the support function that excludes the origin (n=1), which is replaced by the support function projecting from said origin.}
\label{aa:complex_supports}
\end{figure}
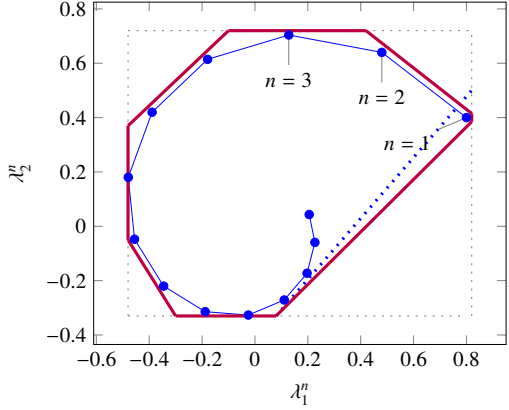

\item Equal Eigenvalues\\
When two eigenvalues are the same, the resulting support functions are those orthogonal to the $x=y$ plane, intersecting the square created by the maximum and minimum values.
\item Jordan Blocks of size $>1$\\
In the case of eigenvalues with geometric multiplicities $>1$, the shape of the function is similar to its corresponding unit size block. In the convergent case, since the convexity can be sharp, it is important to find the apex of the upper diagonals in order to minimise the over-approximation. See Figure~\ref{aa:jordan_supports}.

\begin{figure}[]
\centering
\begin{tikzpicture}[scale=0.8]
\begin{axis}[ylabel={$\lambda_2^n$},xlabel={$\lambda_1^n$}]
\addplot[color=blue,mark=*] table{
0.8 1
0.64 1.6
0.512 1.92
0.4096 2.048
0.3277 2.050
0.2621 1.966
0.2097 1.835
.1678 1.678
.1342 1.510
.1074 1.342
.0859 1.181
.0687 1.031
.0550 0.893
.0440 0.770
.0352 0.660
};
\addplot[dash pattern=on 1pt off 3pt , color=gray] coordinates {(0.01,0.62) (0.82,0.62)};
\addplot[dash pattern=on 1pt off 3pt , color=gray] coordinates {(0.01,2.07) (0.82,2.07)};
\addplot[dash pattern=on 1pt off 3pt , color=gray] coordinates {(0.82,0.62) (0.82,2.07)};
\addplot[dash pattern=on 1pt off 3pt , color=gray] coordinates {(0.01,0.62) (0.01,2.07)};
\addplot[color=purple, line width=1.5pt] coordinates {(0.82,1.02) (0.55,2.07)};
\addplot[color=purple, line width=1.5pt] coordinates {(0.01,.62) (.1,1.5)};
\addplot[color=purple, line width=1.5pt] coordinates {(.1,1.5) (.25,2.07)};
\addplot[color=purple, line width=1.5pt] coordinates {(.25,2.07) (.55,2.07)};
\addplot[color=purple, line width=1.5pt] coordinates {(0.82,1.02) (0.82,0.96)};
\addplot[color=purple, line width=1.5pt] coordinates {(0.02,0.6) (0.82,0.96)};
\addplot[color=purple, line width=1.5pt] coordinates {(0.02,0.6) (0.01,0.62)};
\node at (axis cs:0.8,1) [pin={+180:$n=1$},inner sep=0pt] {};
\node at (axis cs:0.64,1.6) [pin={190:$n=2$},inner sep=0pt] {};
\node at (axis cs:.512,1.92) [pin={-120:$n=3$},inner sep=0pt] {};
\node at (axis cs:.3277,2.05) [pin={-90:$apex$},inner sep=0pt] {};
\end{axis}
\end{tikzpicture}
\caption{Polyhedral faces from an $\mathbb{R}^2$ Jordan block subspace, where 
  $(\lambda_1^n, \lambda_2^n)$ so that $\lambda_1{=}0.8, \lambda_2{=}0.8, 1{\leq}n{\leq}15$. 
  Bold purple lines represent supports found by this article. The blue dotted line shows the support function that excludes the origin (n=1), which is replaced by the support function projecting from said origin.}
\label{aa:jordan_supports}
\end{figure}
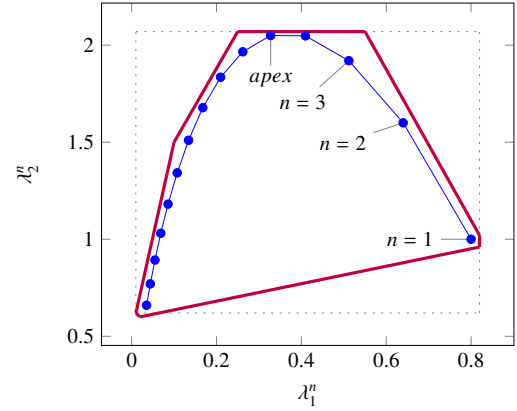

\item Negative Eigenvalues and mixed types\\
When mapping a positive real eigenvalue to a complex conjugate or negative one, we must account for both sides of the axis on the latter. 
These form mirror images that are merged during the abstraction. To make matters simple, we use the magnitude of a complex eigenvalue and evaluate whether the dynamics are concave or convex with respect to the mirroring plane. See Figure~\ref{aa:mixed_supports}.

Note that if both eigenvalues are negative and/or conjugate pairs from a different pair, the mirror image would be taken on both axes, resulting in a hyperrectangle. For a tighter bound in the purely convergent case, we find the convex hull of a point cloud for a small time horizon and merge it with the hyperrectangle for the infinite time horizon thereon.
\end{enumerate}

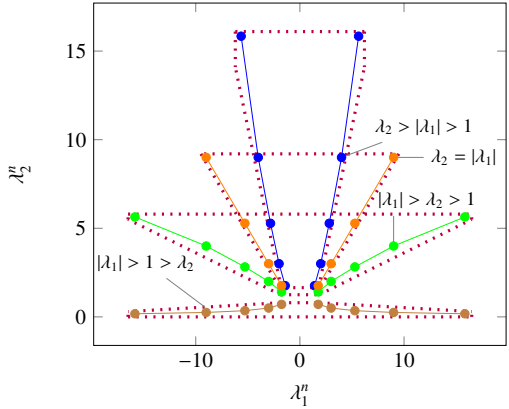
\begin{figure}[]
\centering
\begin{tikzpicture}[scale=0.8]
\begin{axis}[ylabel={$\lambda_2^n$},xlabel={$\lambda_1^n$}]
\addplot[color=blue,mark=*] table{
1.41 1.76
2 3
2.82 5.28
4 9
5.64 15.84
};
\addplot[color=blue,mark=*] table{
-1.41 1.76
-2 3
-2.82 5.28
-4 9
-5.64 15.84
};
\addplot[color=green,mark=*] table{
1.76 1.41
3 2
5.28 2.82
9 4
15.84 5.64
};
\addplot[color=green,mark=*] table{
-1.76 1.41
-3 2
-5.28 2.82
-9 4
-15.84 5.64 
};
\addplot[color=orange,mark=*] table{
1.76 1.76
3 3
5.28 5.28
9 9
};
\addplot[color=orange,mark=*] table{
-1.76 1.76
-3 3
-5.28 5.28
-9 9
};
\addplot[color=brown,mark=*] table{
1.76 .707
3 .5
5.28 .35
9 .25
15.84 .176
};
\addplot[color=brown,mark=*] table{
-1.76 .707
-3 .5
-5.28 .35
-9 .25
-15.84 .176 
};
\addplot[dash pattern=on 1pt off 3pt, color=purple, line width=1.5pt] coordinates {(-1.45,1.65) (1.45,1.65)};
\addplot[dash pattern=on 1pt off 3pt, color=purple, line width=1.5pt] coordinates {(-6.2,16.1) (6.2,16.1)};
\addplot[dash pattern=on 1pt off 3pt, color=purple, line width=1.5pt] coordinates {(-6.2,16.1) (-6.2,14.1)};
\addplot[dash pattern=on 1pt off 3pt, color=purple, line width=1.5pt] coordinates {(-1.8,1.65) (-6.2,14.1)};
\addplot[dash pattern=on 1pt off 3pt, color=purple, line width=1.5pt] coordinates {(6.2,16.1) (6.2,14.1)};
\addplot[dash pattern=on 1pt off 3pt, color=purple, line width=1.5pt] coordinates {(1.8,1.65) (6.2,14.1)};
\addplot[dash pattern=on 1pt off 3pt, color=purple, line width=1.5pt] coordinates {(-1.85,1.25) (1.85,1.25)};
\addplot[dash pattern=on 1pt off 3pt, color=purple, line width=1.5pt] coordinates {(-16.5,5.8) (16.5,5.8)};
\addplot[dash pattern=on 1pt off 3pt, color=purple, line width=1.5pt] coordinates {(1.85,1.25) (16.5,5.5)};
\addplot[dash pattern=on 1pt off 3pt, color=purple, line width=1.5pt] coordinates {(16.5,5.5) (16.5,5.8)};
\addplot[dash pattern=on 1pt off 3pt, color=purple, line width=1.5pt] coordinates {(-1.85,1.25) (-16.5,5.5)};
\addplot[dash pattern=on 1pt off 3pt, color=purple, line width=1.5pt] coordinates {(-16.5,5.5) (-16.5,5.8)};
\addplot[dash pattern=on 1pt off 3pt, color=purple, line width=1.5pt] coordinates {(-1.85,.8) (1.85,.8)};
\addplot[dash pattern=on 1pt off 3pt, color=purple, line width=1.5pt] coordinates {(-16.5,0) (16.5,0)};
\addplot[dash pattern=on 1pt off 3pt, color=purple, line width=1.5pt] coordinates {(1.85,.8) (16.5,.3)};
\addplot[dash pattern=on 1pt off 3pt, color=purple, line width=1.5pt] coordinates {(16.5,0) (16.5,.3)};
\addplot[dash pattern=on 1pt off 3pt, color=purple, line width=1.5pt] coordinates {(-1.85,.8) (-16.5,.3)};
\addplot[dash pattern=on 1pt off 3pt, color=purple, line width=1.5pt] coordinates {(-16.5,0) (-16.5,.3)};
\addplot[dash pattern=on 1pt off 3pt, color=purple, line width=1.5pt] coordinates {(-9.2,9.2) (9.2,9.2)};
\addplot[dash pattern=on 1pt off 3pt, color=purple, line width=1.5pt] coordinates {(-9.5,9.2) (-2.1,1.8)};
\addplot[dash pattern=on 1pt off 3pt, color=purple, line width=1.5pt] coordinates {(9.5,9.2) (2.1,1.8)};
\node at (axis cs:9,4) [pin={+90:\small $\ \ \ \ \ \ \ \ \ \ |\lambda_1|>\lambda_2>1$},inner sep=0pt] {};
\node at (axis cs:4,9) [pin={+30:\small $\lambda_2>|\lambda_1|>1$},inner sep=0pt] {};
\node at (axis cs:9,9) [pin={+0:\small $\lambda_2=|\lambda_1|$},inner sep=0pt] {};
\node at (axis cs:-9,.5) [pin={+98:\small $|\lambda_1|>1>\lambda_2$},inner sep=0pt] {};
\end{axis}
\end{tikzpicture}
\caption{Polyhedral faces from an $\mathbb{R}^2$ subspace, with different convexities  
(note that the blue and orange plots are convex w.r.t. the $\lambda_2^n$-axis, whereas the green and brown are concave).
Dotted purple lines represent supports for some of these layouts.}
\label{aa:mixed_supports}
\end{figure}

An additional drawback of \cite{JSS14} is that calculating the exact
Jordan form of any matrix is computationally expensive and hard to
achieve for large-dimensional matrices.  
We will instead use numerical algorithms in
order to get an approximation of the Jordan normal form
and account for numerical errors. In
particular, if we examine the nature of
\eqref{eq:aa_reach_tube}, we find out that the
numerical operations are not iterative, therefore the errors do not
accumulate with time.  We use properties of
eigenvalues to relax $\vec{f}$ by finding the maximum error in the
calculations that can be determined by computing the norm
$\delta_{max} = |\intmat{A}-\hat{\intmat{S}}\hat{\intmat{J}}\hat{\intmat{S}^{-1}}|$,
where $\hat{\intmat{J}}$ and $\hat{\intmat{S}}$ are the
numerically calculated eigenvalues and eigenvectors of $\mat{A}$.
The notation above is used to represent the matrices as interval
matrices and all operations are performed using interval arithmetic
with outward rounding in order to ensure soundness. In the following
we will presume exact results and use the regular notation to describe the algorithms. 
The constraints $\mat{\Phi} \vec{m} < \vec{f}$ are then computed by
considering the ranges of eigenvalues  $\lambda_s \pm \delta_{max}$
(represented in Fig.~\ref{aa:supports} as the diameter of the blue dots).
The outward relaxation of the support functions ($\vec{f}$), which follows
a principle similar to that introduced in~\cite{gao2012delta}, reduces the
tightness of the over-approximation, but ensures the soundness of the
abstract matrix $\mathcal{A}^n$ obtained. It is also worth noting that
the transformation matrices into and from the eigenspace will also
introduce over-approximations due to the intervals.
One can still use exact arithmetic with a noticeable improvement over
previous work; however, for larger-scale systems the option of using
floating-point arithmetic, while taking into account errors and meticulously
setting rounding modes, provides a $100$-fold plus improvement,
which can make a difference towards rendering verification practically feasible.
For a full description on the numerical processes described here see \cite{cattaruzza2017sound}

\subsection{Abstract Matrices in Support Functions}\label{sec:absmatinsupfunc}

Since we are describing operations using abstract matrices and support
functions, we briefly review the nature of these operations and
the properties that the support functions retain within this domain.  Let $X
\in \mathbb{R}^p$ be a space and $\mathcal{A} \in \mathcal{R}^{p \times p}$
an abstract matrix for the same space.  From the definition we have
\begin{displaymath}
\mathcal{A}=\bigcup \mat{S}\varphi(\vec{m})\mat{S}^{-1} : \mat{\Phi}\vec{m} \leq \vec{f} 
\end{displaymath}
which leads to 
\begin{equation}
\rho_{\mathcal{A}X}(\vec{v})=\rho_{\mat{S}\varphi(\vec{m})\mat{S}^{-1}X}(\vec{v})=\rho_{\varphi(\vec{m})\mat{S}^{-1}X}\left(\mat{S}^T\vec{v}\right), 
\end{equation}
where
\begin{equation}
\rho_{\varphi X}(\vec{v}) = \sup\left \{ \rho_{\varphi}(\vec{x} \circ \vec{v}) \mid \vec{x} \in X\right \}
\end{equation}
and
\begin{equation}
\rho_{\varphi}(\vec{v}) = \sup\{\vec{m} \cdot \varphi^{-1}(\vec{v}) \mid \mat{\Phi}\vec{m} \leq \vec{f}\}
\end{equation}
Here, $\vec{x} \circ \vec{y}$ is the Hadamard product, where $ (\vec{x} \circ \vec{y})_{i}=\vec{x}_i\vec{y}_i$, and $\varphi^{-1}(\cdot)$ is the reverse operation of $\varphi(\cdot)$ in order to align the elements on $\vec{v}$ with the elements in $\vec{m}$. In the case of conjugate pairs this is equivalent to multiplying the vector section by $\tiny \left[\begin{array}{cc}1&1\\1&-1\end{array} \right]$, and in the case of a Jordan block by an upper triangular matrix of all ones.\\

We may also define
\begin{align}
\label{eq:abs_eigensupport}
\rho_{\mathcal{A}X}(\vec{v})=&\sup\left \{\rho_{\mat{a}X}(\vec{v}), \forall \mat{a} \in \mathcal{A}\right \}\nonumber\\
=&\sup\left \{ \mat{S}\varphi(\vec{m})\mat{S}^{-1}\vec{x} \cdot \vec{v},\ \forall \vec{x} \in X\right \}\nonumber\\
=&\sup\left \{ \varphi(\vec{m})\mat{S}^{-1}\vec{x} \cdot \mat{S}^T\vec{v},\ \forall \vec{x} \in X\right \}\nonumber\\
=&\sup\left \{\rho_{\varphi}(\mat{S}^{-1}\vec{x} \circ \mat{S}^T\vec {v}),\ \forall \vec{x} \in X\right \}. 
\end{align}
In order to simplify the nomenclature we write
\begin{equation}
\rho_{\mathcal{A}X}(\vec{v})=\rho_{X}(\mathcal{A}^T \vec {v}).
\end{equation}

\subsection{Acceleration of Parametric Inputs}

Let us now consider the following over-approximation for $\tau$ on sets:
\begin{equation}
\tau^\sharp(X_0,U) = \mat{A}X_0\oplus\mat{B}U \;
\end{equation}
Unfolding~\eqref{equ:reachset} (ignoring the presence of the guard set $G$ for the time being), we obtain 
$$ 
X_n=\mat{A}^nX_0\oplus\sum_{k\in [0,n-1]} \mat{A}^k\mat{B}U \;
$$%
What is left to do is to further simplify the sum $\sum_{k\in [0,n-1]}
\mat{A}^k\mat{B}U$.
We can exploit the following simple results from linear algebra. 
\begin{lemma}
\label{eq::accel}
If $\mat{I}-\mat{A}$ is invertible, then 
$$\sum\limits_{k=0}^{n-1} \mat{A}^k = (\mat{I} -\mat{A}^{n})(\mat{I} - \mat{A})^{-1}$$.  
If furthermore $\lim\limits_{n\to \infty} \mat{A}^n=0$, then 
$$\lim\limits_{n \to \infty} \sum\limits_{k=0}^n \mat{A}^k = (\mat{I}-\mat{A})^{-1}$$. 
\end{lemma}

This lemma presents a difficulty in the nature of $\mat{A}$. The inverse $(\mat{I}-\mat{A})^{-1}$,
does not exist for eigenvalues of $1$, i.e. we need $1\notin \sigma(\mat{A})$, 
where $\sigma(A)$ is the spectrum (the set of all the eigenvalues) of matrix $\mat{A}$.
In order to overcome this problem, we introduce the eigen-decomposition of $\mat{A} = \mat{S}\mat{J}\mat{S}^{-1}$, 
(and trivially $\mat{I}=\mat{S}\mat{I}\mat{S}^{-1}$), and by the distributive and transitive properties we obtain 
\begin{align*}
 (\mat{I} -\mat{A}^{n})(\mat{I} - \mat{A})^{-1} & = \mat{S}(\mat{I}-\mat{J}^{n})(\mat{I}-\mat{J})^{-1}\mat{S}^{-1} \;. 
\end{align*}
This allows us to accelerate the eigenvalues individually, using the property $\sum_{k=0}^{n-1} 1^k = n$ for eigenvalues of $\lambda=1$.
Using the properties above, and translating the problem into the generalised eigenspace to accounting for
unit eigenvalues, we obtain the following representation:
\begin{align}
\displaystyle
\sum_{k=0}^{n-1} \lambda^k &= 
\left \{
\begin{array}{cc}
 n & \lambda =1 \\
\frac{1-\lambda^n}{1-\lambda} & \lambda \neq 1\\
\end{array}
\right. \nonumber\\
&\Rightarrow
(\mat{I}-\mat{A}^n)(\mat{I}-\mat{A})^{-1}= \mat{S}\mat{D}^n\mat{S}^{-1} \nonumber\\
d(\lambda_i,n,k)&=\frac{-1^k}{k+1}\frac{1-\lambda_i^n}{(1-\lambda_i)^{k+1}}\nonumber\\
&+\sum\limits_{j=1}^{k}\frac{-1^{k-j}}{k-j}\binom{n}{j-1}\frac{\lambda_i^{n-j-1}}{(1-\lambda_i)^{k-j}}\nonumber\\
\mat{D}^n_{i,j}&=\left\{
\begin{array}{cc}
0 & j < i\\
n & i=j \wedge \lambda_i = 1\\
\frac{1-\lambda_i^n}{1-\lambda_i} & i=j \wedge \lambda_i \neq 1 \\
0 & gm(\lambda_i)=1\\
\binom{n+1}{k+1} & \lambda_i = 1 \\
d(\lambda_i,n,j-i) & \lambda_i \neq 1 \\
\end{array} \right.
\label{abs:bover}
\end{align}
where $gm(\cdot)$ is the geometric multiplicity of the given eigenvalue.

\subsection{Acceleration of Variable Inputs}\label{sec:var_inpur_accel}

The result in the previous section can be only directly applied under restricted conditions in the case of variable inputs. For
instance whenever $\forall k>0,\ \vec{u}_k=\vec{u}_{k-1}$. 
In order to generalise it (in particular to non-constant inputs), we will
over-approximate $\mat{B}U$ over the eigenspace by a semi-spherical enclosure with
centre $\vec{u}_c'$ and radius $U_b'$.  To this end, we first rewrite
\begin{displaymath}
U_J' = \mat{S}^{-1}\mat{B}U = \{ \vec{u}_c'\} \oplus U_d', 
\end{displaymath}
where $\vec{u}_c'$ is the centre of the interval hull of $U_J'$:
\begin{displaymath}
{\mat{u}_c'}_i = \frac{1}{2}(\rho_{U_J'}(\vec{v_i})+\rho_{U_J'}(-\vec{v_i})) \mid {\vec{v}_i}_j= \left \{ 
\begin{array}{cc}
1 & j=i\\
0 & j \neq i
\end{array}
.\right.
\end{displaymath}
We then over-approximate $U_d'$ via $U_b'$, by the maximum radius in
the directions of the complex eigenvalues as (cf. illustration in Figure~\ref{aa:varinputs}). Let

\begin{align*}
\Lambda&=\{\lambda_i \mid i\in[1, p], \lambda_i^*\neq \lambda_{i-1}\}\\
f_b&:\mathbb{R}^p \rightarrow \mathbb{R}^{p_b} \text{ such that}\\
f_b(\vec{v})&=red(\vec{v}_b) \text{ where } (\vec{v}_b)_i=\left \{ \begin{array}{cc}
0 & \lambda_i \notin \Lambda\\
|\vec{v}_i| & \lambda_i \neq \lambda_{i+1}^*\\
\sqrt{\vec{v}_i^2+\vec{v}_{i+1}^2} & \lambda_i = \lambda_{i+1}^* \end{array} \right.
\end{align*}
and red($\cdot$) is a function that reduces the dimension of a vector by removing the elements where $ \lambda_i \notin \Lambda$.
Extending this to matrices we have
\begin{align}
F_b&: \mathbb{R}^{o \times p} \rightarrow \mathbb{R}^{o \times p_b}\nonumber\\
F_b(\mat{C})&=\mat{C}_b \text{ where }  ({\mat{C}_b})_{i,*}=f_b(\mat{C}_{i,*})
\end{align}
Finally
\begin{align}
U_d'&=\{ \vec{u} \mid \mat{C}_u'\vec{u} \leq \vec{d}_u' \}\nonumber\\
U_d' &\subseteq U_b' = \{\vec{u} \mid F_b(\mat{C}_u')f_b(\vec{u}) \leq f_b(\vec{d}_u') \}\nonumber\\
\mat{B}U &\subseteq U_b \oplus U_c \mid U_b=\mat{S}U_b' \text{ and } U_c=\{\mat{S}\vec{u}_c'\}
\label{eq:overball}
\end{align}

Since the description of $U'_b$ is no longer polyhedral in $\mathbb{R}^p$,
we will also create a semi-spherical over-approximation $\mat{J}_b$ of
$\mat{J}$ in the directions of the complex eigenvectors, in a
similar way as we generated $U_b'$ for $U_d'\ $. 
More precisely,
\begin{align}
\label{eq:round_jordan}
\mat{J}_b&=\left [ \begin{array}{ccc}
\mat{J}_{b1} &          & \\
           & \ddots & \\
           &           & \mat{J}_{br}^n\\
\end{array} \right ]\text{ where }\nonumber\\
\forall s \in [1\ r]& \left \{\begin{array}{l}\lambda_j \in {\mat{J}_b}_s = |\lambda_i| \in \mat{J}_s \cap \Lambda\\ gm({\mat{J}_b}_s)=gm(\mat{J}_s)\end{array}\right.
\end{align}
where $gm(\cdot)$ is the geometric multiplicity of the Jordan block.
\begin{definition}
Given a matrix $\mat{A}=\mat{S}\mat{J}\mat{S}^{-1}$ and a vector $\vec{x}$, we define the following operations:
\begin{align}
\label{eq: over_mult}
F_b^*(\mat{A},\vec{x})&=\mat{S}f_b^{-1}\left(F_b(\mat{J})f_b(\mat{S}^{-1}\vec{x})\right)\\
F_b'(\mat{A},\vec{x})&=f_b^{-1}\left(F_b(\mat{J})f_b(\vec{x}')\right)
\label{eq: over_eigen_mult}
\end{align}
\end{definition}
Finally, we refer to the accelerated sets 
\begin{align*}
U_b^n&=\left\{F_b^*((\mat{I}-\mat{A}^n),F_b^*((\mat{I}-\mat{A})^{-1},\vec{u})) \mid \vec{u} \in U_b \right\}\\
U_c^n&=(\mat{I}-\mat{A}^n)(\mat{I}-\mat{A})^{-1}U_c\\
U_{cb}^n&=U_c^n \oplus U_b^n
\end{align*}

Returning to our original equation for the $n$-reach set, we obtain%
\footnote{Note that $\forall U_b'\ ,U_c'\
,U_d'\ ;\ \exists U_b\ ,U_c\ ,U_d : U_b'=\mat{S}^{-1}BU_b$ so that $U_c'=\mat{S}^{-1}BU_c$ and $U_d'=\mat{S}^{-1}BU_d$.  
Hence, this inclusion is also valid in the original state space. }
\begin{align}
\label{equ:accinput2}
X_n \subseteq \mat{A}^nX_0 &\oplus U_{cb}^n
\end{align}

\begin{figure}
\centering
\begin{tikzpicture}[scale=0.9]
\begin{axis}[ylabel={$\vec{\lambda}_2$},xlabel={$\vec{\lambda}_1$},xmin=-50,xmax=50,ymin=-40,ymax=40]
\addplot[color=blue,mark=*] table{
40.6 23
37.1 24.75
5.95 24.75
4.95 22.75
-4.95 -0.25
8.45 -8
39.6 -2
40.6 0
40.6 23
};
\addplot[dashed, color=blue,mark=*] table{
22.7 11.63
19.2 13.38
-12.05 13.38
-13.05 11.38
-22.85 -11.62
-10.35 -19.37
21.7 -13.37
22.7 -11.37
22.7 11.63
};
\addplot[color=purple,mark=*] table{
17.9 11.37
0 0
};
\draw[color=red] (axis cs:17.9,11.37) circle[radius=260];
\draw[dashed,color=orange] (axis cs:0,0) circle[radius=260];
\node at (axis cs:17.9,11.37) [pin={-10:$\vec{u}_c'$},inner sep=0pt] {};
\node at (axis cs:0,0) [pin={-10:$0,0$},inner sep=0pt] {};
\node at (axis cs:12.1,-23.37) [pin={-10:$U_b'$},inner sep=0pt] {};
\node at (axis cs:8,13) [pin={+10:$U_d'$},inner sep=0pt] {};
\node at (axis cs:32,-4) [pin={-10:$U'$},inner sep=0pt] {};
\end{axis}
\end{tikzpicture}
\caption{Relaxation of an input set within a complex subspace, 
in order to make it invariant to matrix rotations.  
Dashed lines and curves denote translated quantities onto the origin. }
\label{aa:varinputs}
\end{figure}
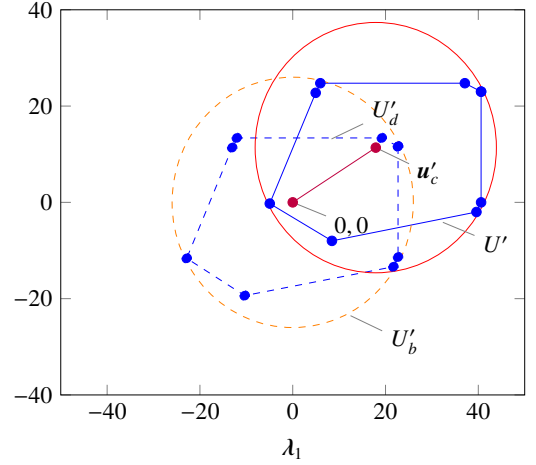

Shifting our attention from reach sets to tubes, 
we can now over-approximate the {\it reach tube} by abstract acceleration of the three summands in (\ref{equ:accinput2}), as follows. 
\begin{theorem}
The abstract acceleration
\begin{equation}
\tau^{\sharp n}(X_0,U) =_{\mathrm{def}} \mathcal{A}^nX_0 \oplus \mathcal{B}^n U_c \oplus \mathcal{B}_b^n U_b
\label{equ:absaccinput}
\end{equation}
is an over-appro\-xi\-mation of the $n$-reach tube, 
namely $\hat{X}_n \subseteq \tau^{\sharp n}(X_0,U)$.
\end{theorem}
\begin{proof}
The proof is derived from that in~\cite{JSS14} for $\mathcal{A}^nX_0$, 
and extends it as in the developments presented above.
\end{proof}

\subsection{Combining Abstract Matrices}\label{sec:comb_matrices}

One important property of the abstract matrices $\mathcal{A}^n$,
$\mathcal{B}^n$ and $\mathcal{B}_b^n$ is that they are correlated.  In~the case of parametric
inputs, this correlation is linear and described by the acceleration defined
in Lemma \eqref{eq::accel}.  In the case of $\mathcal{B}_b^n$ this
relationship is not linear (see Eq.~\ref{eq:overball}).  However, we can
still find a linear over-approximation of the correlation between
$\mathcal{B}_b^n$ and $\mathcal{A}^n$ based on the time steps $k$.  Given
two orthonormal spaces $X \in \mathbb{R}^p \wedge U \in \mathbb{R}^q$ and a
transition equation
\begin{displaymath}
X_{k+1}=\mat{A}X_k + \mat{B}U, 
\end{displaymath}
which is related to 
\begin{displaymath}
 \rho_{X_{k+1}}(\vec{v})=\rho_{\mat{A}X_k}(\vec{v}) + \rho_{\mat{B}U}(\vec{v}), 
\end{displaymath}
we define a space 
\begin{displaymath}
X'=\left\{ \cvec{\vec{x}}{\mat{B}\vec{u}} \mid \vec{x} \in X, \vec{u} \in U \right\}
\end{displaymath}
so that 
\begin{displaymath}
\rho_{X_{k+1}}(\vec{v})=\rho_{X'_k}\cvec{\mat{A}^T\vec{v}}{\vec{v}}=\rho_{X'_k}\left (\mat{D}^T\vec{v}'\right), 
\end{displaymath}
with 
\begin{displaymath}
\mat{D}=\qmat{\mat{A}}{0}{0}{\mat{I}}, \quad \vec{v}'= \cvec{\vec{v}}{\vec{v}}. 
\end{displaymath}
Accelerating $X_{k+1}$, we obtain 
\begin{displaymath}
\rho_{X_{n}}(\vec{v})=\rho_{\mat{A}^nX_0}(\vec{v}) + \rho_{(\mat{I}-\mat{A}^{n})(\mat{I}-\mat{A})^{-1}\mat{B}U}(\vec{v})=\rho_{X'_0}\left (\mat{D}^{nT}\vec{v}'\right), 
\end{displaymath}
with 
\begin{displaymath}
\mat{D}^n=\qmat{\mat{A}^n}{0}{0}{(\mat{I}-\mat{A}^n)(\mat{I}-\mat{A})^{-1}}
\end{displaymath}
in the case of parametric inputs.  More generally, the diagonal elements of
$\mat{D}^n$ correspond to the diagonal elements of $\mat{A}^n$ and
$\sum_{k=0}^{n-1}{\mat{A}^k} \mat{B}$,  which means we can construct
\begin{equation}
\mathcal{D}^n=\qmat{\mathcal{A}^n}{0}{0}{\mathcal{B}^n} \mid \rho_{X_n}(\vec{v})=\rho_{X'_0}(\mathcal{D}^{nT}\vec{v}'). 
\end{equation}
We can then apply this abstraction to (\ref{eq:overball}) and obtain:
\begin{align}
\label{eq:varacc}
\rho_{X_n}(\vec{v})&=\rho_{X'_0}(\mathcal{D}_b^{nT}\vec{v}') \text{ where }\\
\mathcal{D}_b^n&=\qmat{\mathcal{A}^n}{0}{0}{\mathcal{B}_b^n}\ ,\ \ \vec{v}'= \cvec{\vec{v}}{f_b(\vec{v})} \nonumber\\
 \mathcal{B}_b^n &= \mat{S}F_b^{-1}\left((\mat{I}-\mathcal{J}_b^n)(\mat{I}-\mat{J}_b)^{-1}F_b(\mat{S}^{-1})\right)\nonumber
\end{align}
with $\mat{J}_b$ defined by \eqref{eq:round_jordan}. 
This model provides a tighter over-approximation than \eqref{equ:absaccinput} since the accelerated dynamics of the inputs are now constrained by the accelerated dynamics of the system.

\section{Abstract Acceleration with Guards: Estimation of the number of Iterations} \label{sec:guards}

The most important task remaining is how to calculate the number of iterations dealing with the presence of the guard set $G$. 

Given a convex polyhedral guard expressed as the assertion $\{\vec{x} \mid \mat{G}\vec{x} \leq \vec{h}\}$, 
we define $G_{i,*}$ as the $i^{th}$ row of $\mat{G}$ and $h_i$ as the corresponding
element of $\vec{h}$. 
We denote the normal vector to the $i^{th}$ face of the guard as $\vec{g}_i=G_i^T$.  
The distance of the guard to the origin is thus $\gamma_i = \frac{h_i}{|\vec{g}_i|}$.

Given a convex set $X$, we may now describe its position with respect to
each face of the guard through the use of its support function alongside the normal
vector of the hyperplane (for clarity, we assume the origin to be inside set $X$): 
\begin{displaymath}
\begin{array}{ccl}
\rho_X(\vec{g}_i) \leq \gamma_i, & & \textrm{inside\ the\ hyperplane},\\
-\rho_X(-\vec{g}_i) \geq \gamma_i, & & \textrm{outside\ the\ hyperplane}. 
\end{array}
\end{displaymath}
Applying this to equation \eqref{equ:accinput2} we  obtain:
\begin{align}
& \label{equ:underiter} \rho_{X_n}(\vec{g}_i) =\rho_{X_0}({\mat{A}^{\underline{n_i}}}^T\vec{g}_i)+\rho_{U_{cb}^n}(\vec{g}_i) \leq \gamma_i \\
& \label{equ:overiter}  \rho_{X_n}(-\vec{g}_i)=\rho_{X_0}(-{\mat{A}^{\overline{n_i}}}^T\vec{g}_i)+\rho_{U_{cb}^n}(-\vec{g}_i) \leq -\gamma_i 
\end{align}

From the inequalities above we can determine up to which number of iterations $\underline{n_i}$ the reach tube remains inside the corresponding hyperplane, 
and starting from which iteration $\overline{n_i}$ the corresponding reach set goes beyond the guard:

In order for a reach set to be inside the guard it must therefore be inside
all of its faces, and we can ensure it is fully outside of the guard set
when it is fully beyond any of them.  Thus, we have $\underline{n} = \min\{\
\underline{ n_i }\ \}$, and $\overline{n} = \min \{\ \overline{ n_i }\ \}$.

We have not however discussed why these two cases are important. Looking at the transition in equation \eqref{equ:reachtraj}, we can easily derive that if $\mat{G}\vec{x}_k \not \leq \vec{h}$ the postimage of all subsequent iterations is empty. Therefore, any overapproximation henceforth will only add imprecision. We will use the bounds $\underline{n}$ and $\overline{n}$ to create a tighter overapproximation.
Let 
\begin{align}
\hat{X}_{\underline{n}}^{\sharp}&=\mathcal{A}_{\underline{n}}X_0\oplus\mathcal{B}_{\underline{n}}U\ \ \ \ \text{  (n-reachtube)}\nonumber\\
X_{\underline{n}}^{\sharp}&=\mat{A}_{\underline{n}}X_0\oplus\mathcal{B}_{\underline{n}}U\ \ \ \ \text{  (n-reachset)}\nonumber\\
\hat{X}_{\overline{n}\ \mid\ \underline{n}}^{\sharp}&=\tau\left(\mathcal{A}_{\overline{n}-\underline{n}-1}X_{\underline{n}}^{\sharp}\oplus\mathcal{B}_{\underline{n}}U\cap G, U\right)\nonumber\\
\hat{X}_{\overline{n}}^{\sharp}&=\hat{X}_{\overline{n}\ \mid\ \underline{n}}^{\sharp} \cup \hat{X}_{\underline{n}}^{\sharp}
\end{align}
This double step prevents the set $\left\{ \vec{x} \mid \vec{x} \in \hat{X}_{\underline{n}}^{\sharp}, \vec{x} \notin X_{\underline{n}}^{\sharp} \right\}$ to be included in further projections, thus reducing the size of the overapproximation. 

\medskip 
Computing the maximum $\underline{n_i}$ such that \eqref{equ:underiter} is satisfied is not easy 
because the unknown $\underline{n_i}$ occurs in the exponent of the equation.
However, since an intersection with the guard set will always return a
sound over-approximation, we do not need a precise value. We can over-approximate it by decomposing $\vec{g}_i$ into the generalised eigenspace of $\mat{A}$.
Let $\vec{g}_i=\sum_{j=1}^p k_{ij} \vec{v}_j +\res(\vec{g}_i)$, 
where $\vec{v}_j$ are row vectors of $\mat{S}^{-1}$ or $-\mat{S}^{-1}$ such that $k_{ij} \geq 0$, 
and $\res(\vec{g}_i)$ is the component of $\vec{g}_i$ that lies outside the range of $\mat{S}$.  
Notice that since $\mat{S}$ has an inverse, it is full rank and therefore
$\res(\vec{g}_i)=\vec{0}$ and subsequently not relevant.  It is also
important to note that $\mat{S}$ is the matrix of generalised eigenvectors
of $\mat{A}$ and therefore we are expressing our guard in the generalised
eigenspace of $\mat{A}$.
\begin{equation}
\rho_{X_0}({\mat{A}^n}^T \vec{g}_i) = \rho_{X_0}\left(\sum_{j=1}^p k_{ij} {\mat{A}^n}^T \vec{v}_j\right) \leq \sum_{j=1}^p k_{ij} \rho_{X_0}\left({\mat{A}^n}^T \vec{v}_j\right)
\end{equation}

\subsection{Overestimating the Iterations of a loop without inputs} \label{sec:guards_noinputs}
We start by looking into the approximation of the inside bound (i.e. the iterations for which the
reachtube remains fully inside the guard). Since rotating dynamics and Jordan shapes will have
a complex effect on the behaviour of the equation, we seek to transform the Jordan form into a
real positive diagonal. In such a case, the progression of the support function in each direction
is monotonically increasing (or decreasing) and it is therefore very easy to find a bound for its 
progression. We note that the envelope of rotating dynamics will always contain the true
dynamics and is therefore a sound overapproximation. We will initially assume that 
$\gamma_i$ is positive and then extend to the general case.

Let $\rho_{X_0}({\mat{A}^n}^T \vec{g}_i)=\rho_{X_0'}({\mat{J}^n}^T \vec{g}_i')$ such that
\begin{align}
\vec{g}_i'&=\mat{S}^{-1}\vec{g}_i\nonumber\\
X_0&=\{ \vec{x} \mid \mat{C}_{X_0}\vec{x} \leq \vec{d}_{X_0} \}\nonumber\\
X_0'&=\mat{S}^{-1}X_0=\{ \vec{x} \mid \mat{S}\mat{C}_{X_0}\vec{x} \leq \vec{d}_{X_0} \}\nonumber
\end{align}


Let
\begin{align*}
\Lambda_\sigma&=\{\lambda_i : i\in[1,p],\ \bigwedge_{j=1}^{i-1} (\lambda_i^*\neq \lambda_j \wedge \lambda_i\neq \lambda_j)\}\\
f_\sigma(\vec{v})&:\mathbb{R}^p \rightarrow \mathbb{R}^{p_b}\\
f_\sigma(\vec{v})&=red(\vec{v}_\sigma) \\
\text{ where }& ({\vec{v}_\sigma})_i=
\left \{ \begin{array}{cc}
0 & \lambda_i \notin \Lambda_\sigma\\
\sqrt{\sum\limits_{j\in[1, p]\wedge  (\lambda_j=\lambda_i \vee \lambda_j=\lambda_i^*) }{\vec{v}_j^2 }}& \lambda_i \in \Lambda_\sigma
\end{array} \right.\\
F_\sigma&: \mathbb{R}^{o \times p} \rightarrow \mathbb{R}^{o \times r}\\
F_\sigma(\mat{C})&=\mat{C}_\sigma \text{ where }  ({\mat{C}_\sigma})_{i,*}=f_\sigma(\mat{C}_{i,*}).
\end{align*}

and red($\cdot$) is a function that reduces the dimension of a vector by removing the elements where $ \lambda_i \notin \Lambda_\sigma$. This reduction is not extrictly necessary, but it enables a faster implementation by reducing dimensionality.
Correspondingly, given $\mat{J}=diag\left(\{ \mat{J}_s \mid s \in [1, r]\}\right)$
\begin{align}
\label{eq:overJb}
\mat{J}_\sigma&=\left[\begin{array}{cccc}\overline{\sigma}_1&0&\cdots&0\\0&\overline{\sigma}_2&\cdots&0\\0&0&\ddots&0\\0&0&\cdots&\overline{\sigma}_r\\\end{array}\right] 
\end{align}
where $\overline{\sigma}_s=||\mat{J}_s||_2$ is the maximum singular value (hence the induced norm ~\cite{LT84}) of the Jordan block $\mat{J}_s$.

Finally, let
\begin{align}
\vec{x}_{c}'&=\frac{1}{2}(\rho_{X_0'}(\vec{v_i})+\rho_{X_0'}(-\vec{v_i})), \quad {\vec{v}_i}_j= \left \{ 
\begin{array}{cc}
1 & j=i\\
0 & j \neq i
\end{array}
\right.
\nonumber\\
X_{\sigma}' &= \{\vec{x} \mid F_\sigma(\mat{S}\mat{C}_{X_0})f_\sigma(\vec{x}) \leq f_\sigma(\vec{d}_{X_0}-\mat{S}\mat{C}_{X_0}\vec{x}_c') \}\nonumber\\
X_0' &\subseteq f_\sigma^{-1}(X_{c\sigma}') \mid X_{c\sigma}' =\{f_\sigma(\vec{x}_c')\} \oplus X_{\sigma}' 
\label{eq:overball_init}
\end{align}
and $\vec{v}_\sigma=f_\sigma(\vec{v})$.

Using eigenvalue and singular value properties, we
obtain $\rho_{X_0}({\mat{A}^n}^T \vec{v}_j) \leq
\overline{\sigma}_j^{\ n}\,\rho_{X_{c\sigma}}((\vec{v}_\sigma)_j) \mid j\in [1, r]$, and
therefore:
\begin{equation}
\rho_{X_0}({\mat{A}^n}^T \vec{g}_i) \leq \sum_{j=1}^p k_{ij} \overline{\sigma}_j^n\,\rho_{X_{c\sigma}}((\vec{v}_\sigma)_j))
\end{equation}

Since 
we have no inputs, $\rho_{U_c^n}(\vec{g}_i) + \rho_{U_b^n}(\vec{g}_i)=0$, hence we may solve for $\underline{n}_i$:
\begin{align}
\label{eq:low_n}
\rho_{X_0}({\mat{A}^{\underline{n}_i}}^T \vec{g}_i) &\leq \sum_{j=1}^p k_{ij} \overline{\sigma}_j^{\ \underline{n}_i}\,\rho_{X_{c\sigma}}((\vec{v}_\sigma)_j) \leq \gamma_i
\end{align}

To separate the divergent element of the dynamics from the
convergent one, let us define
\begin{align*}
\overline{k}_{ij}&=max \left( k_{ij}\ \rho_{X_{c\sigma}}((\vec{v}_\sigma)_j)\ ,\ 0 \right)\\
 \overline{\sigma}&=max\left(\{ \overline{\sigma}_s \mid s \in [1, p] \}\right).
\end{align*}
This step will allow us to track effectively which trajectories are
likely to hit the guard and when, since it is only the divergent
element of the dynamics that can increase the reach tube in a given
direction.

Replacing \eqref{eq:low_n}, we obtain
\begin{equation}
 \overline{\sigma}^{\ n}\sum_{j=1}^p \overline{k}_{ij}\left(\frac{ \overline{\sigma}_j}{ \overline{\sigma}}\right)^n \leq
\gamma_i\;, 
\end{equation}
which allows to finally formulate an iteration scheme for approximating $n$.

\begin{proposition}\label{prop:real_under_n}
  An iterative under-approximation of the number of iterations $n$ can
  be computed by starting with $\underline{n_i}=0$ and iterating over
\begin{equation}
\underline{n_i} \geq n=
\log_{ \overline{\sigma}}\left({\gamma_i}\right)-\log_{ \overline{\sigma}}\left({\sum_{j=1}^p \overline{k}_{ij}\left(\frac{ \overline{\sigma}_j}{ \overline{\sigma}}\right)^{\underline{n_i}}}\right)
\;,  
\label{eq:est_low_n}
\end{equation}
substituting $n_i=n$ on the right-hand side until we meet the inequality. 
%
\end{proposition}
\begin{proof}
This follows from the developments unfolded above. 
Notice that the sequence $\underline{n_i}$ is monotonically increasing, before it breaks the inequality. 
As such any local minimum represents a sound under-approximation of the number of loop iterations. 
Note that in the case where $\gamma_i \leq 0$ we must first translate the system coordinates such that $\gamma_i > 0$. 
This is simply done by replacing $\vec{x}'=\vec{x}+\vec{c}$ and operating over the resulting system where  $\gamma_i' = \rho_{\vec{c}}(\vec{g}_i)+\gamma_i$.

Mathematically this is achieved as follows: first we get $\vec{c}$ by finding the center of the interval hull of $G$ (if $G$ is open in a given direction we may pick any number in that direction for the corresponding row of $c$). Next we transform the dynamics into
\begin{equation*}
\left[\begin{array}{c}\vec{x}_{k}\\ \mat{1}\end{array}\right]=
\left[\begin{array}{cc}\mat{A}&\mat{A}\vec{c}\\\mat{0}&\mat{1}\end{array}\right] \left[\begin{array}{c}\vec{x}_{k-1}\\ \mat{1}\end{array}\right]
+\left[\begin{array}{c}\mat{B}\\ \mat{0}\end{array}\right]\vec{u}_{k} \mid \left[\begin{array}{c}\vec{x}_{k-1}\\ \mat{1}\end{array}\right] \in G'
\end{equation*}
where 
\begin{equation*}
G'=\left\{\left[\begin{array}{c}\vec{x}\\ \mat{1}\end{array}\right] \mid \left[\begin{array}{cc}\mat{G}&\mat{G}\vec{c}\\\mat{0}&\mat{1}\end{array}\right] \left[\begin{array}{c}\vec{x}_{k-1}\\ \mat{1}\end{array}\right] \leq \left[\begin{array}{c}\vec{h}\\ \mat{1}\end{array}\right]\right\}
\end{equation*}
\end{proof}

\subsection{Underestimating the Iterations of a loop without inputs} \label{sec:guards_under}
In order to apply a similar techniques to \eqref{equ:overiter} we must find an equivalent under-approximation. In the case of equation \eqref{eq:low_n}, the $\overline{\sigma}_j$ esure that the equation diverges faster than the real dynamics, hence the iteration found is an upper bound to the desired iteration. In this case we want the opposite, hence we look for a model where the dynamics diverge slower. In this case it is easy to demonstrate that ${\lambda_b}_j=|\lambda_j|$ represents these slower dynamics.
\begin{align}
\rho_{X_0}(-{\mat{A}^{\overline{n_i}}}^T\vec{g}_i) &\leq \sum_{j=1}^p k_{ij} {\lambda_b}_j^{\ \overline{n_i}}\,\rho_{X_{c\sigma}}(-(\vec{v}_\sigma)_j) \leq -\gamma_i
\label{eq:high_n}
\end{align}
which reduces to
\begin{equation}
 \overline{\sigma}^{\ n}\sum_{j=1}^p \underline{k}_{ij}^-\left(\frac{ {\lambda_b}_j}{ \overline{\sigma}}\right)^n +\overline{\sigma}^{\ n}\sum_{j=1}^p \underline{k}_{ij}^+ \leq -\gamma_i\;, 
\end{equation}
where
\begin{align*}
\underline{k}_{\ ij}^-&=min \left( k_{ij}\ \rho_{X_{c\sigma}}(-(\vec{v}_\sigma)_j)\ ,\ 0 \right)\\
\underline{k}_{\ ij}^+&=max \left( k_{ij}\ \rho_{X_{c\sigma}}(-(\vec{v}_\sigma)_j)\ ,\ 0 \right)
\end{align*}

An additional consideration must also be made regarding the rotational nature of the dynamics. In the previous case we did not care about the rotational alignment of the set $X_n$ with respect to the vector $\vec{g}_i$, because any rotation would move the set inside the guard. In this case, although the magnitude of the resulting vector is greater than the required one, the rotation may cause it to be at an angle that keeps the set inside the guard. 
We must therefore account for the rotating dynamics in order to find the point where the angles align with the guard. In order to do this, let us first fix the magnitudes of the powered eigenvalues, in the case of convergent dynamics we will assume they have converged a full rotation in  to make our equation strictly divergent. Let $\underline{\theta}=\min \{\theta_j \mid j \in [1, p]\}$, where $\theta_j$ are the angles of the complex conjugate eigenvalues. Let $n_\theta=\frac{2\pi}{\underline{\theta}}$ be the maximum number of iterations needed for any of the dynamics to complete a full turn. Then at any given turn $|\lambda_j|^{\overline{n}_i+n_\theta} \leq |\lambda_j|^{\overline{n}_i+n} \mid |\lambda_i| \leq 1, n \in 0 n_\theta$.
This means that any bound we find on the iterations will be necessarily smaller than the true value. Our problem becomes the solution to:
\begin{align*}
&\max\left( \overline{\sigma}^{\ \overline{n}_i}\sum_{j=1}^p c_{ij} cos((n-\overline{n}_i) \theta_j - \alpha_{ij}) \right)\\
\alpha_{ij}&=\cos^{-1}(\vec{g}_i \cdot \vec{v}_j)\\
c_{ij}&=\left \{\begin{array}{ll}
\underline{k}_{ij}^-\left(\frac{ {\lambda_b}_j}{ \overline{\sigma}}\right)^{\overline{n}_i}&|\lambda_j| \geq 1 \\
\underline{k}_{ij}^-\left(\frac{ {\lambda_b}_j}{ \overline{\sigma}}\right)^{\overline{n}_i+n_\theta}&|\lambda_j| < 1
\end{array} \right.
\end{align*}
The problem is simplified by underapproximating the cosines and removing the constants:
\begin{align*}
&\max\left( \overline{\sigma}^{\ \overline{n}_i}\sum_{j=1}^p c_{ij} \left(1-\frac{((n-\overline{n}_i) \theta_j - \alpha_{ij})^2}{2}\right) \right)\\
\Rightarrow&\min\left( \sum_{j=1}^p c_{ij} {((n-\overline{n}_i) \theta_j - \alpha_{ij})^2} \right)\\
\Rightarrow&\min\left( \sum_{j=1}^p c_{ij}\theta_j^2(n-\overline{n}_i)^2 + c_{ij}\alpha_{ij}\theta_j (n-\overline{n}_i) \right)\\
\end{align*}
The solution to this equation is 
\begin{equation}
n=\overline{n}_i-\frac{\sum_{j=1}^p c_{ij}\alpha_{ij}\theta_j}{2\sum_{j=1}^p c_{ij}\theta_j^2} \mid n \in [\overline{n}_i, \overline{n}_i+n_\theta]
\label{eq:circ_iters}
\end{equation}
The second part of the equation is expected to be a positive value. When this is not the case, the dominating dynamics will have a rotation $\theta_j \geq \frac{pi}{2}$. In such cases we must explicitly evaluate the set of up to $4$ iterations after $\overline{n}_i$.
If the resulting bound does not satisfy the original inequality: $\rho_{X_0}\left({\mat{A}^{\overline{n_i}}}^T \vec{g}_i\right) \geq\gamma_i$, we replace $\overline{n}_i=n$ until it does \footnote{this is a tighter value than work shown on previous versions of this paper where we overapproximated using $n_\theta=\frac{(2\pi)^m}{\prod_j \theta_j}$, where $m$ is the number of conjugate pairs.}.

\begin{proposition}\label{prop:real_under_n}
  An iterative under-approximation of the number of iterations $n$ can
  be computed by starting with $\overline{n_i}'=0$ and iterating over
\begin{align}
\overline{n_i}' &\leq n=
\log_{ \overline{\sigma}}\left({\gamma_i}\right)-\log_{ \overline{\sigma}}\left({\sum_{j=1}^p \underline{k}_{ij}^-\left(\frac{ {\lambda_b}_j}{ \overline{\sigma}}\right)^{\overline{n_i}'}}+\sum\limits_{j=1}^p \underline{k}_{ij}^+\right)\nonumber\\
\overline{n_i}&=\overline{n_i}'+k \mid \rho_{X_0}\left({\mat{A}^{(\overline{n_i}'+k)}}^T \vec{g}_i\right) \geq\gamma_i
\label{eq:est_high_n}
\end{align}
where $k$ is the result of equation \eqref{eq:circ_iters}.
we substitute for $\overline{n}_i=n$ on the right-hand side until we break the inequality, and then find $k$ such that the second inequality holds.
\end{proposition}
Since we are explicitly verifying the inequality, there is no further proof required.
\subsection{Estimating the Iterations of a loop with inputs} \label{sec:guards_inputs}
For the case with inputs, we will use the same paradigm explained in the previous section after performing a mutation that transforms the system with inputs into an over-approximating system without inputs.

Let $X_{c\sigma}', U_{c\sigma}'$ be the corresponding sets of initial states and inputs obtained by applying equation \eqref{eq:overball_init} to $X_0'$ and $U_J'$, and let $U_{J\sigma}'=(\mat{I}-\mat{J}_\sigma)^{-1} U_{c\sigma}'$. The accelerated resulting system may be represented by the equations
\begin{align}
(X_{c\sigma}')_n&=\mat{J}_\sigma^n X_{c\sigma}' \oplus (\mat{I}-\mat{J}_\sigma^n)U_{J\sigma'}\nonumber\\
\rho_{(X_{c\sigma}')_n}(\vec{v})&=\rho_{X_{c\sigma}'}\left(\mat{J}_\sigma^{nT}\vec{v}\right) +\rho_{U_{J\sigma}'}(\vec{v}) -\rho_{U_{J\sigma}'}\left(\mat{J}_\sigma^{nT}\vec{v}\right)
\end{align}
Let us now define $(XU)_{\sigma}=\{\vec{x}-\vec{u} \mid \vec{x} \in X_{c\sigma}', \vec{u} \in U_{J\sigma}' \}$ which allows us to translate the system into
\begin{equation}
\rho_{((XU)_{\sigma}')_n}(\vec{v})=\rho_{(XU)_{\sigma}'}\left(\mat{J}_\sigma^{nT}\vec{v}\right)
\end{equation}
which has the same shape as the equations in the previous section. We may now apply the techniques described above to find the bounds on the iterations.

\rronly{
\subsection{Narrowing the estimation of the iterations} \label{sec:guards_tight}
The estimations above are very conservative, but we may use further techniques to obtain tighter bounds on the number of iterations.
In the first instance we note that we have eliminated all negative terms in the sums in equation \eqref{eq:est_low_n}. Reinstating these terms can cause us to lose monotonicity, but we may still create an iterative approach by fixing the negative value at intermediate stages.
Let $\underline{n}_i$ be our existing bound for the time horizon before reaching a guard, and $\underline{k}_{\underline{n}_i}=\sum_{j=1}^p \underline{k}_{ij}\left(\frac{ \overline{\sigma}_j}{ \overline{\sigma}}\right)^{\underline{n}_i}$,  $\overline{k}_{\underline{n}_i}=\sum_{j=1}^p \overline{k}_{ij}\left(\frac{ \overline{\sigma}_j}{ \overline{\sigma}}\right)^{\underline{n}_i}$ the corresponding negative and positive terms of the equation.
We may now find upper and lower bounds for $\underline{n}_i$ by replacing the equation 
\begin{equation}
\underline{n_i} \geq n_k=\log_{ \overline{\sigma}}\left({\gamma_i}\right)-\log_{\overline{\sigma}}\left(\overline{k}_{\underline{n}_i}+\underline{k}_{\underline{n}_k}\right)
\label{eq:refine_iters}
\end{equation}
where $\underline{n}_k$ is the bound found in the previous stage. Some stages of this process will provide an unsound result, but they will also provide an upper bound to our number of iterations. 
In fact, every second stage will provide a monotonically increasing sound bound which will be tighter than the one in equation \eqref{eq:est_low_n}.
\begin{proof}
Since the elements of the sums are convergent, we have
\begin{align*}
n_i \geq n_k &\Rightarrow \underline{k}_{n_i} \geq \underline{k}_{n_k} \left(\text{i.e. $|\underline{k}_{n_i}| \leq |\underline{k}_{n_k}|$}\right) \\
& \Rightarrow \log_{\overline{\sigma}}\left(\overline{k}_{\underline{n}_i}+\underline{k}_{\underline{n}_k}\right) \geq \log_{\overline{\sigma}}\left(\overline{k}_{\underline{n}_i}+\underline{k}_{\underline{n}_i}\right)
\end{align*}
which means that $n_k$ in equation \eqref{eq:refine_iters} is smaller than or $n$ in equation \eqref{eq:est_low_n} ($n_k \leq n \leq \underline{n_i} \mid \underline{n_i} \geq \underline{n_k})$.
\end{proof}

In the case of equation \eqref{eq:est_high_n}, the explicit evaluation of the guard at each cycle executes the behaviour described here.
}
\subsection{Maintaining Geometric Multiplicity} \label{sec:guards_geometric}
A second step in optimising the number of iterations comes from adding granularity to the bounding semi-spherical abstraction by retaining the geometric multiplicity using the matrix $\mat{J}_b$.

\begin{lemma}
\label{lemma:jordan_guards}
Given a matrix $\mat{A}$ with eigenvalues $\{\lambda_s \mid s\in [1, r]\}$, where each eigenvalue $\lambda_s$ has a geometric multiplicity $p_s$ and corresponding generalised eigenvectors $\{\vec{v}_{s,i} \mid i \in [1, p_s]\}$,
\begin{align}
\forall n \geq 0, \mat{A}^n\vec{v}_{s}^{i}&=\lambda_{s}^n\vec{v}_{s,i}+\sum_{j=1}^{i-1} \lambda_{s}^{n-j}\prod\limits_{k=0}^{j-1} (n-k) \vec{v}_{s,i-j}\nonumber\\
&=\lambda_{s}^n\left( \vec{v}_{s,i} + \sum_{j=1}^{i-1} \frac{\prod\limits_{k=0}^{j-1} (n-k) }{\lambda_{s}^j}\vec{v}_{s,i-j}\right)
\label{eq:jordan_iters}
\end{align}  
\end{lemma}
\begin{proof}
By definition, given an eigenvector $\vec{v}_s$ of $\mat{A}$, then $\mat{A}\vec{v}_s=\lambda_s\vec{v}_s$~\cite{horn2012matrix}. Similarly a generalised eigenvector $\vec{v}_{s,i}$ of $\mat{A}$ satisfies the equation $\left(\mat{A}-\lambda_s\mat{I}\right)\vec{v}_{s,i}=\vec{v}_{s,i-1}$ and $\vec{v}_{s,1}=\vec{v}_s$ hence 
\begin{align*}
\mat{A}\vec{v}_{s,i}&=\lambda_s\vec{v}_{s,i}+\vec{v}_{s,i-1}\\
\mat{A}^n\vec{v}_{s,1}&=\lambda_s^n\vec{v}_{s,1}\\
\mat{A}^n\vec{v}_{s,i}&=A^{n-1}(\lambda_s\vec{v}_{s,i}+\vec{v}_{s,i-1})=\lambda_sA^{n-1}\vec{v}_{s,i}+\mat{A}^{n-1}\vec{v}_{s,i-1}\\
&=\lambda_s^2A^{n-2}\vec{v}_{s,i}+\lambda_sA^{n-2}\vec{v}_{s,i-1}+\mat{A}^{n-1}\vec{v}_{s,i-1}\\
&=\cdots=\lambda_s^n\vec{v}_{s,i}+\sum_{j=0}^{n-1}\lambda_s^j\mat{A}^{n-j-1}\vec{v}_{s,i-1}
\end{align*}
From here we recursively expand the formula for $\mat{A}^{n-j-1}\vec{v}_{s,i-1}$ and obtain:
\begin{align*}
\mat{A}^n\vec{v}_{s,i}&=\lambda_s^n\vec{v}_{s,i}+\sum_{j=0}^{n-1}\lambda_s^j\lambda_s^{n-j-1}\vec{v}_{s,i-1}+\sum_{j=0}^{n-1}\sum_{k=0}^{n-2}\lambda_s^k\mat{A}^{n-k-2}\vec{v}_{s,i-2}\\
&=\lambda_s^n\vec{v}_{s,i}+n\lambda_s^{n-1}\vec{v}_{s,i-1}+n\sum_{j=0}^{n-2}\lambda_s^j\mat{A}^{n-j-2}\vec{v}_{s,i-2}\\
&=\cdots=\lambda_{s}^n\vec{v}_{s,i}+\sum_{j=1}^{i-1} \lambda_{s}^{n-j}\prod\limits_{k=0}^{j-1} (n-k) \vec{v}_{s,i-j}
\end{align*}
\end{proof}

Let $i'$ denote the position of $f_b(\lambda_{j})$ within the block $\mat{J}_{bs}$ it belongs to, such that its corresponding generalised eigenvector is identified as $\vec{v}_{bs,i'}=f_b(\vec{v}_j)$.
Then
\begin{align}
&\rho_{X_0'}({\mat{J}^n}^T \vec{g}_i')\nonumber\\
&\leq \sum_{j=1}^{p_b} k_{ij} \rho_{X_0}\left({\mat{J}_b^n}^T f_b(\vec{v}_j)\right)\nonumber\\
&\leq \sum_{j=1}^{p_b} k_{ij} {\lambda_b}_j^n \rho_{X_0}\left( \vec{v}_{bs,i'} + \sum_{k=1}^{i'-1} \frac{\prod\limits_{m=0}^{k-1} (n-m) }{{\lambda_b}_j^k}\vec{v}_{bs,i'-k}\right)\nonumber\\
&\leq \sum_{j=1}^{p_b} k_{ij} {\lambda_b}_j^n \left(\rho_{X_0}\left( \vec{v}_{bs,i'}\right) + \sum_{k=1}^{i'-1} \frac{\prod\limits_{m=0}^{k-1} (n-m) }{{\lambda_b}_j^k}\rho_{X_0}\left(\vec{v}_{bs,i'-k}\right)\right)\nonumber\\
&\leq \sum_{j=1}^{p_b} k_{ij0}' {\lambda_b}_j^n+ \sum_{m=1}^{i'} k_{ijm}' {\lambda_b}_j^n\prod\limits_{m=0}^{p_s-i'-1} (n-m)
\end{align}

I order to manage the product on the right hand side we use slightly different techniques for over- and under-approximations.
For $\underline{n}_i$ we first find an upper bound $\underline{n}_i'$ using equation \eqref{eq:est_low_n} and $k_{ij}={k_{ij}'}_0+{k_{ij}'}_m$
and then do a second iteration using $k_{ij}={k_{ij}'}_0+{k_{ij}'}_m\prod\limits_{m=0}^{p_s-i'-1} (\underline{n}_i'-m)$ which ensures the true value is under the approximation.
In the case of $\overline{n}_i$, we also start with $k_{ij}={k_{ij}'}_0+{k_{ij}'}_m$ and update it during the iterative process.

Let us look at the following example:
\begin{align*}
\mat{J}&\scriptsize
=\left[\begin{array}{cccccc}
3&0&0& 0& 0&0\\
0&2&1& 0& 0&0\\
0&0&2& 0& 0&0\\
0&0&0&-2& 0&0\\
0&0&0& 0& -1&1\\
0&0&0& 0&-1&-1
\end{array}\right]\\
\mat{S}&\scriptsize
=\left[\begin{array}{cccccc}
1&0&0&0& 0&0\\
0&3&0&0& 0&0\\
0&-4&1&0& 0&0\\
0&0&0&1& 0&0\\
0&0&0&0& 1&0\\
0&0&0&0& 1&1
\end{array}\right]\\
\end{align*}
\begin{align*}
\mat{J}_{\sigma}&\scriptsize
=\left[\begin{array}{cccc}
3&0&0& 0\\
0&3&0& 0\\
0&0&2& 0\\
0&0&0&\sqrt{2}
\end{array}\right]
\end{align*}
\begin{align*}
\vec{x}_0'&={\scriptsize\left[\begin{array}{cccccc}1&1&1&1&1&1\end{array}\right]}\\
\mat{G}\vec{x} \leq 300 &\mid \mat{G}{\scriptsize=\left[\begin{array}{cccccc}1&3&-3&2&4&1\end{array}\right]}\\
\mat{G}&={\scriptsize \left[\begin{array}{cccccc}1&1&1&2&4&-3\end{array}\right]}\mat{S}^T
\end{align*}

The progression of the system along the support function and corresponding bounds as described in the previous section are shown in figure~\ref{fig:iters}

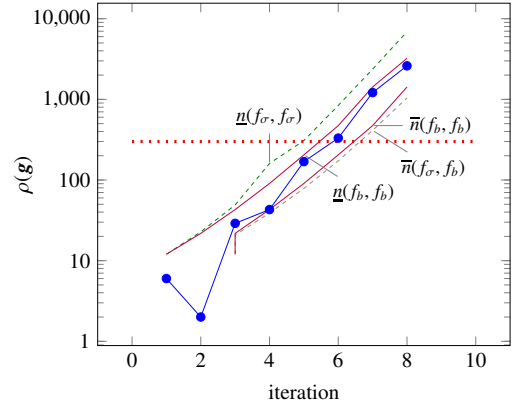
\begin{figure}[]
\centering
\begin{tikzpicture}[scale=0.8]
\begin{semilogyaxis}[ylabel={$\rho(\vec{g})$},xlabel={iteration}, log ticks with fixed point]
\addplot[color=blue,mark=*] table{
1 6
2 2
3 29
4 43
5 169
6 331
7 1213
8 2603
};
\addplot[dash pattern=on 2pt off 2pt, color=green!50!black,mark=.] table{
1 12
2 22.8
3 49
4 160.6
5 308
6 840.2
7 2371
8 6895.6
};
\addplot[dash pattern=on 2pt off 2pt, color=yellow!50!black,mark=.] table{
3 12
3 20.8
4 39
5 78.6
6 173
7 412.6
8 1041
};
\addplot[color=purple,mark=.] table{
1 12
2 21.8
3 43
4 90.6
5 205
6 472.6
7 1425
8 3225.6
};\addplot[color=purple,mark=.] table{
3 12
3 21.8
4 43
5 90.6
6 205
7 472.6
8 1425
};]
\addplot[dash pattern=on 1pt off 3pt, color=red, line width=1.5pt] coordinates {(0,300) (10,300)};
\node at (axis cs:4,160) [pin={+90:\small $\underline{n}(f_\sigma,f_\sigma)$},inner sep=0pt] {};
\node at (axis cs:7,412) [pin={-45:\small $\overline{n}(f_\sigma,f_b)$},inner sep=0pt] {};
\node at (axis cs:5,205) [pin={-45:\small $\underline{n}(f_b,f_b)$},inner sep=0pt] {};
\node at (axis cs:7,473) [pin={0:\small $\overline{n}(f_b,f_b)$},inner sep=0pt] {};
\end{semilogyaxis}
\end{tikzpicture}
\caption{Progression of the support function of a system for a given guard. Blue dots are real values.
The dashed green line overapproximates the progression using singular values (sec~\ref{sec:guards_noinputs}), 
the dashed yellow line underapproximates them using eigenvalue norms (sec~\ref{sec:guards_under}), 
whereas the continuous purple lines represent the tighter overapproximation maintaining the gemoetric 
multiplicity (sec~\ref{sec:guards_geometric}). We can see how the purple line finds a
better bound for $\underline{n}_i$, while the $\overline{n}_i$ bound is conservative for both approaches. Mind the logarithmic scale.}
\label{fig:iters}
\end{figure}

Changing the eigenvalues to:
\begin{align*}
\mat{J}&
=\scriptsize \left[\begin{array}{cccccc}
2e^{-0.2i}&0&0& 0& 0&0\\
0&2e^{0.2i}&0& 0& 0&0\\
0&0&\sqrt{2}e^{-0.3i}& 0& 0&0\\
0&0&0&\sqrt{2}e^{0.3i}& 0&0\\
0&0&0& 0& 1.1e^{0.5i}&0\\
0&0&0& 0&0&1.1e^{-0.5i}
\end{array}\right]
\end{align*}

we get the results in figure~\ref{fig:round_iters}. In this case we can see that the rotational dynamics force an increase of the initially calculated iteration to account for the effects of the rotation.

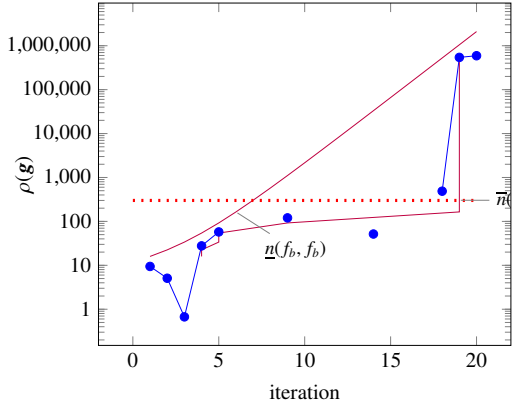
\begin{figure}[]
\centering
\begin{tikzpicture}[scale=0.8]
\begin{semilogyaxis}[ylabel={$\rho(\vec{g})$},xlabel={iteration}, log ticks with fixed point]
\addplot[color=blue,mark=*] table{
1 9.43
2 5.08
3 0.67
4 27.72
5 57.58
};
\addplot[color=blue,mark=*] table{
6 -7.79
7 -171.74
8 -197.26
};
\addplot[color=blue,mark=*] table{
9 120.35
};
\addplot[color=blue,mark=*] table{
10 -20.24
11 -2778.33
12 -8156.45
13 -8912.76
};
\addplot[color=blue,mark=*] table{
14 51.58
};
\addplot[color=blue,mark=*] table{
15 -8545.78
16 -86240.20
17 -175680.75
};
\addplot[color=blue,mark=*] table{
18 487.48
19 542552.47
20 591848.64
};
\addplot[color=purple,mark=.] table{
1 15.94
2 22.47
3 33.80
4 54.25
5 92.24
6 164.40
7 303.58
8 575.00
9 1108.38
10 2162.15
11 4251.72
12 8405.95
13 16679.66
14 33178.53
15 66108.22
16 131872.05
17 263265.32
18 525862.65
19 1050790.65
20 2100270.50
};
\addplot[color=purple,mark=.] table{
4 15.94
4 22.47
5 33.80
5 54.25
9 92.24
19 164.40
19 303.58
19 575.00
19 1108.38
19 2162.15
19 4251.72
19 8405.95
19 16679.66
19 33178.53
19 66108.22
19 131872.05
19 263265.32
19 525862.65
};
\addplot[dash pattern=on 1pt off 3pt, color=red, line width=1.5pt] coordinates {(0,300) (20,300)};
\node at (axis cs:6,164) [pin={-45:\small $\underline{n}(f_b,f_b)$},inner sep=0pt] {};
\node at (axis cs:19,303) [pin={0:\small $\overline{n}(f_b,f_b)$},inner sep=0pt] {};
\end{semilogyaxis}
\end{tikzpicture}
\caption{Progression of the support function of a rotational system for a given guard. Blue dots are real values 
(negative values are missing due to the log scale).
Continuous purple lines represent the overapproximation. The steep vertical line at 19 is due to the alignment of the rotations with the guard at this point. The point at iteration 14 appears below the line because of the higher point at iteration 9. The model will either find that this boundary was met at iteration 9 or push it forward to 19.}
\label{fig:round_iters}
\end{figure}

\subsection{Case Study}
We have selected a known benchmark to illustrate the discussed procedure:  
the room temperature control problem \cite{Fehnker04benchmarksfor}. 
The temperature (variable \texttt{temp}) of a room is controlled to a user-defined set point (\texttt{set}), 
which can be changed at any time through a heating (\texttt{heat}) element, 
and is affected by ambient temperature (\texttt{amb}) that is out of the control of the system. 

We formalise the description of such a system both via a linear loop and via hybrid dynamics.  
Observe that since such a system may be software controlled, 
we assume that part of the system is coded, 
and further assume that it is possible to discretise the physical environment for simulation. 
Algorithm \ref{alg:temp_code} shows a pseudo-code fragment for the temperature control problem.
\begin{algorithm}
\caption{Temperature Control Loop}
\textbf{States:} temp=temperature, heat=heat output.\\
\textbf{Inputs:} set=set-point, amb=ambient temperature.
\begin{algorithmic}[1]
\State temp=5+read(35);
\State heat=read(1);
\State while(temp$<400$ $\&\&$ heat$<300$)
\State \{
\State     \hspace{.5cm} amb=5+read(35);
\State     \hspace{.5cm} set=read(300);
\State     \hspace{.5cm} temp=.97 temp + .02 amb + .1 heat;
\State     \hspace{.5cm} heat=heat + .05 set; 
\State \}
\end{algorithmic}
\label{alg:temp_code}
\end{algorithm}
We use the \texttt{read} function to represent non-deterministic values
between 0 and the maximum given as argument.  Alternatively, this loop
corresponds to the following hybrid dynamical model:
\begin{align*}
\left [ \begin{array}{c}
temp\\
heat\\
\end{array}
\right ]_{k+1}
&=
\left [ \begin{array}{cc}
0.97 & 0.1\\
-0.05 & 1\\
\end{array}
\right ]
\left [ \begin{array}{c}
temp\\
heat\\
\end{array}
\right ]_{k}\\
&+
\left [ \begin{array}{cc}
0.02 & 0\\
0 & 0.05\\
\end{array}
\right ]
\left [ \begin{array}{c}
amb\\
set\\
\end{array}
\right ]_{k}, 
\end{align*}
with initial condition
\begin{align*}
\left [ \begin{array}{c}
temp\\
heat\\
\end{array}
\right ]_{0}
\in
\left [ \begin{array}{c}
\left [ 5\ \ 40 \right ]\\
\left [ 0\ \  1 \right ]\\
\end{array}
\right ],
\end{align*}
non-deterministic inputs
\begin{align*}
\left [ \begin{array}{c}
amb\\
set\\
\end{array}
\right ]_{k}
\in
\left [ \begin{array}{c}
\left [ 5\ \ 40 \right ]\\
\left [ 0\ \ 300 \right ]\\
\end{array}
\right ],
\end{align*}
and guard set
\begin{align*}
G = \left\{
\left [ \begin{array}{c}
temp\\
heat\\
\end{array} \right ] \mid 
\left [ \begin{array}{cc}
1 & 0\\
0 & 1\\
\end{array}
\right ]
\left [ \begin{array}{c}
temp\\
heat\\
\end{array}
\right ]
<
\left [ \begin{array}{c}
400\\
300\\
\end{array}
\right ] \right\}.
\end{align*}

In this model the variables are continuous and take values over the real line, 
whereas within the code they are represented as long double precision floating-point values, with precision of $\pm 10^{-19}$, 
moreover the error of the approximate Jordan form computation results in $\delta_{max}<10^{-17}$. 
Henceforth we focus on the latter description, as in the main text of this work.
The eigen-decomposition of the dynamics is (the values are rounded to three decimal places):
\begin{align*}
\tiny
A&=SJS^{-1} \subseteq \intmat{S}\intmat{J}\intmat{S}^{-1}\text{ where }\\
\intmat{S}&=\left [ \begin{array}{cc}
0.798 \pm 10^{-14} & 0.173 \pm 10^{-15}\\
0 \pm 10^{-19} & 0.577\pm 10^{-14}\\
\end{array}
\right ]\\
\intmat{J}&=
\left [ \begin{array}{cc}
0.985 \pm 10^{-16}& 0.069 \pm 10^{-17}\\
-0.069 \pm 10^{-17}& 0.985 \pm 10^{-16}\\
\end{array}
\right ]\\
\intmat{S}^{-1}&=
\left [ \begin{array}{cc}
1.253 \pm 10^{-12} & - 0.376\pm 10^{-13}\\
0 \pm 10^{-18} & 1.732 \pm 10^{-12}\\
\end{array}
\right ]. 
\end{align*}
The discussed over-approximations of the reach-sets indicate that the
temperature variable intersects the guard at iteration $\underline{n}=32$. 
Considering the pseudo-eigenvalue matrix (described in the extended version for the case of complex eigenvalues)
along these iterations, we
use Equation~\eqref{abs:mat} to find that the corresponding complex pair
remains within the following boundaries:
\begin{displaymath}
\small
\begin{array}{lr}
\mathcal{A}^{32} = 
\left [
\begin{array}{cc}
r & i \\
-i & r\\
\end{array}
\right ]
\left \{
\begin{array}{ccccc}
0.4144 &<& r     &< & 0.985\\
0.0691 &<&  i     &< &0.7651\\
0.1082 &<& r+ i &< &1.247\\
0.9159 &<& i - r &< &0.9389\\
\end{array}
\right.\;\;\\
\mathcal{B}^{32} = 
\left [
\begin{array}{cc}
r & i \\
-i & r\\
\end{array}
\right ]
\left \{
\begin{array}{ccccc}
1 &<& r     &< & 13.41\\
0 &<&  i     &< &17.98\\
1 &<& r+ i &< &29.44\\
6.145 &<& i - r &< &6.514\\
\end{array}
\right.
\end{array}
\end{displaymath} 

The reach tube is calculated by multiplying these abstract matrices with the
initial sets of states and inputs, as described in
Equation~\eqref{equ:absaccinput}, by the following inequalities:
\begin{align*}
\small
\hat{X}_{32}^\#=& \mathcal{A}^{32} 
\left [ \begin{array}{c}
\left [5\ \ 40 \right ]\\
\left [0\ \ \ \ 1 \right ]\\
\end{array}
\right ]
 + \mathcal{B}^{32}
\left [ \begin{array}{c}
\left [5\ \ \ 40 \right ]\\
\left [0\ 300 \right ]\\
\end{array}
\right ]\\
=&
\small
\left [
\begin{array}{c}
temp\\
heat\\
\end{array}
\right ]
\left \{
\begin{array}{ccccc}
-24.76 &<& temp     &< & 394.5\\
-30.21 &<&  heat     &< &253\\
-40.85 &<& temp+ heat &< &616.6\\
-86.31 &<& temp - heat &< &843.8\\
\end{array}
\right.
\end{align*}
The negative values represent the lack of restriction in the code on the lower side and correspond to system cooling (negative heating). 
The set is displayed in Figure~\ref{fig:abs}, where for the sake of clarity we display only 8 directions of the 16 constraints. 
This results in a rather tight over-approximation that is not much looser than the convex hull of all reach sets obtained by \cite{FLD+11} using the given directions. 
In Figure~\ref{fig:abs}, we can see the initial set in black colour, 
the collection of reach sets in white, 
the convex hull of all reach sets in dark blue (as computed by \cite{FLD+11}), 
and finally the abstractly accelerated set in light yellow (dashed lines). 
The outer lines represent the guards. 

\begin{figure}[]
\centering
\begin{tikzpicture}[scale=0.7]
\begin{axis}[
height=0.55\textwidth,
ylabel={$heat$},
xlabel={$temp$},]
\addplot[color=red,domain = -20:420] {300};
\node at (axis cs:0,300) [pin={-10:$heat=300$},inner sep=0pt] {};
\addplot[color=red] coordinates {(400,-20) (400,310)};
\addplot[dashed,color=gray,fill=yellow] table{
393.548 235.724
384.908 253.004
384.908 253.004
384.908 253.004
223.005 253.004
110.38 196.691
4.42517 90.7365
-26.1104 29.6655
-26.1104 -13.1998
-19.9766 -25.4673
-17.9272 -27.5167
-12.3942 -30.2832
58.4534 -30.2832
246.894 63.9373
384.232 201.275
393.548 219.906
393.548 235.724
};
\addplot[color=blue,fill=blue] table{
390.191 219.721
387.012 226.079
384.546 228.545
361.243 240.197
227.71 240.197
103.086 177.885
39.4593 114.258
-19.2963 -3.253
-19.2963 -17.3902
-16.1173 -23.7483
-13.6509 -26.2146
-7.92582 -29.0772
10.0934 -29.0772
272.694 102.223
390.191 219.721
390.191 219.721
390.191 219.721
};
\addplot[red] table{./Reach32.dat};
\addplot[blue,fill=white] table{./Reach31.dat};
\addplot[blue,fill=white] table{./Reach30.dat};
\addplot[blue,fill=white] table{./Reach29.dat};
\addplot[blue,fill=white] table{./Reach28.dat};
\addplot[blue,fill=white] table{./Reach27.dat};
\addplot[blue,fill=white] table{./Reach26.dat};
\addplot[blue,fill=white] table{./Reach25.dat};
\addplot[blue,fill=white] table{./Reach24.dat};
\addplot[blue,fill=white] table{./Reach23.dat};
\addplot[blue,fill=white] table{./Reach22.dat};
\addplot[blue,fill=white] table{./Reach21.dat};
\addplot[blue,fill=white] table{./Reach20.dat};
\addplot[blue,fill=white] table{./Reach19.dat};
\addplot[blue,fill=white] table{./Reach18.dat};
\addplot[blue,fill=white] table{./Reach17.dat};
\addplot[blue,fill=white] table{./Reach16.dat};
\addplot[blue,fill=white] table{./Reach15.dat};
\addplot[blue,fill=white] table{./Reach14.dat};
\addplot[blue,fill=white] table{./Reach13.dat};
\addplot[blue,fill=white] table{./Reach12.dat};
\addplot[blue,fill=white] table{./Reach11.dat};
\addplot[blue,fill=white] table{./Reach10.dat};
\addplot[blue,fill=white] table{./Reach9.dat};
\addplot[blue,fill=white] table{./Reach8.dat};
\addplot[blue,fill=white] table{./Reach7.dat};
\addplot[blue,fill=white] table{./Reach6.dat};
\addplot[blue,fill=white] table{./Reach5.dat};
\addplot[blue,fill=white] table{./Reach4.dat};
\addplot[blue,fill=white] table{./Reach3.dat};
\addplot[blue,fill=white] table{./Reach2.dat};
\addplot[blue,fill=white] table{./Reach1.dat};
\addplot[blue,fill=black] table{./Reach0.dat};
\end{axis}
\end{tikzpicture}
\caption{The abstractly accelerated tube (yellow, dashed boundary),
  representing an over-approximation of the thermostat reach tube
  (dark blue).  The set of initial conditions is shown in black, whereas
  successive reach sets are shown in white.  The guards and the
  reach set that crosses them are close to the boundary in red.
} \label{fig:abs}
\end{figure}
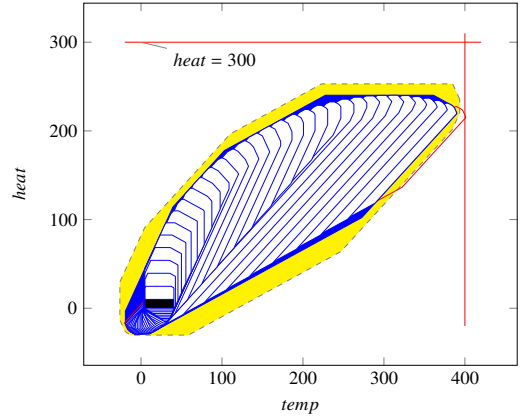

\jrronly{
\section{Applying Abstraction Refinements to Abstract Acceleration}\label{sec:aa_CEGAR}

One of the main limitations of Abstract Acceleration is that despite being very fast, it leverages
on finding an over-approximation of the actual reach-tube for verification. While in many cases this 
over-approximation is suitable, it can occasionally be too conservative for the proof of the properties being sought.
This section deals with methods for refining this over-approximation while remaining fast. In all cases, 
the refinement is based on a counterexample of a support not meeting the specification. This means that
our approach is a Counterexample- Guided Abstraction Refinement (CEGAR) loop for which we use different 
models to obtain the refinements.

\subsection{Finding Counterexample Iterations}\label{sec:cegar_iters}
In all cases, because the objective is to refine the abstract dynamics, we need to find the corresponding 
iteration to the counterexample which will allow us to reduce the polyhedron in the right direction. Since the Abstract 
Dynamics are built over pairs of eigenvalues, it is possible that multiple iterations are found for a given 
refinement, in which case all of them are used.

Let a verification step explore the solution $\rho_\mathcal{A}(\vec{v}) = s < \overline{s}$, where 
$\vec{v}$ is the direction we are examining, $s$ is its corresponding support function, and $\overline{s}$
the safety specification. If $s>\overline{s}$ the specification will not be met and we need a refinement.
Let $\vec{a}_v \in \mathcal{A}$ be the vertex at which the maximum was found, i.e., $\vec{a}_v \cdot \vec{v}=s$.
We can discover a set of possible iterations reaching $\vec{a}_v$ by analysing its location with respect to
the dynamics in the following form:
\begin{enumerate}
\item Conjugate eigenvalues\\
Since the trajectories along these are circular and centred at the origin, we 
can find the angle that $\vec{a}_v$ forms with these axes and use it to calculate the right iteration. Let 
$\theta_i$ be the angle of the conjugate eigenvalue pair and $\theta_{\vec{a}_v}(i)$ the angle formed
by $\vec{a}_v$ in the $i^th$ plane (since we are using pseudo-eigenspaces, this is equivalent to 
$tan^{-1}\left(\frac{(\vec{a}_v)_i}{(\vec{a}_v)_{i+1}}\right)$. 
The corresponding iteration will depend on whether the eigenvalue is convergent or divergent. In the former
case, it will be $\frac{\theta_{\vec{a}_v}(i)}{\theta_i}$, and in the latter it will be $(n - (n \mod \frac{1}{\theta_i}))+\frac{\theta_{\vec{a}_v}(i)}{\theta_i}$,
where $\mod$ is the modulus operation over the reals.
\item Real Eigenvalues\\
In the case of reals, the solution relies on the direct relation between the given
eigenvalue and the target counterexample. Since $(\vec{a}_v)_i \approx \lambda_i^k \Rightarrow k \approx log_{\lambda_i}((\vec{a}_v)_i)$. 
If the logarithm does not exist, then we presume we cannot further refine using this method.
\item Jordan Blocks with non-unitary gemoetric multiplicity\\
In the case of larger Jordan blocks we need to examine the nature of the dynamics. Let us look at the equation representing the contribution of a Jordan block to the support:
\begin{equation}
\rho_{\lambda_s}(\vec{v})=\sum_{j=0}^{p_s} \binom{n}{j} \lambda^{n-j} {\vec{v}_s}_j
\end{equation}
In this case we must use an iterative approximation as described in section \ref{sec:guards_geometric} to find the closest iteration to the breaking guard. Although this processis more costly than the ones described above, it is also more precise, thus providing a much better refinement. Note as well, that the technique can be applied to the full set of eigenvalues or to any subset of Jordan blocks. This choice is a compromise between precision and speed. We also note that when the refinement process is done in the eigenspace, the new eigenvectors are now the identity set, which makes the problem more tractable. 
\end{enumerate}
Since the exclusion of an unsafe vertex from the Abstract Dynamics does not ensure the tightest over-approximation, 
we must perform this step iteratively until either we run out of new refinements or have a user-defined timeout.
Once the candidate iterations are found, it suffices to add further constraints to the abstract matrix for these iterations as described in figure~\ref{fig:cegar_jordan_supports}.
We may also note that given the above procedure, it is often faster and more beneficial to begin by 
performing the refinement over the complex eigenvalues by directly examining the directions vector $\vec{v}$ 
in the corresponding sub-spaces.
\begin{figure}[]
\centering
\begin{tikzpicture}
\begin{axis}[ylabel={$\lambda_2^n$},xlabel={$\lambda_1^n$},scale=.9]
\addplot[color=purple, fill=yellow, line width=1pt] coordinates{
(0.01,.62)
(0.1,1.5)
(0.25,2.07)
(0.42,2.07)
(0.61,1.84)
(0.82,1.02)
(0.82,0.96)
(0.02,0.6)
(0.01,0.62)
};
\addplot[color=blue,mark=*] coordinates{
(0.8,1)
(0.64,1.6)
(0.512,1.92)
(0.4096,2.048)
(0.3277,2.050)
(0.2621,1.966)
(0.2097,1.835)
(.1678,1.678)
(.1342,1.510)
(.1074,1.342)
(.0859,1.181)
(.0687,1.031)
(.0550,0.893)
(.0440,0.770)
(.0352,0.660)
};
\addplot[color=purple, line width=1pt] coordinates{
(0.01,.62)
(.1,1.5)
(.25,2.07)
(.55,2.07)
(0.82,1.02)
(0.82,0.96)
(0.02,0.6)
(0.01,0.62)
};
\addplot[dash pattern=on 1pt off 3pt, line width=1.5pt , color=blue] coordinates {(0.45,2.1) (0.82,1.75)};
\addplot[color=red,mark=*,only marks] coordinates {
(0.55,2.07)
};
\addplot[color=green, line width=1.5pt] coordinates {(0.42,2.07) (0.61,1.84)};
\node at (axis cs:0.8,1) [pin={+180:$n=1$},inner sep=0pt] {};
\node at (axis cs:0.64,1.6) [pin={190:$n=2$},inner sep=0pt] {};
\node at (axis cs:.512,1.92) [pin={-120:$n=3$},inner sep=0pt] {};
\node at (axis cs:.512,1.92) [color=red,pin={-120:$\textcolor{red}{n=3}$},inner sep=0pt] {};
\node at (axis cs:.3277,2.05) [pin={-90:$apex$},inner sep=0pt] {};
\node at (axis cs:0.55,2.07) [pin={0:$violation$},inner sep=0pt] {};
\node at (axis cs:0.78,1.78) [pin={-90:$guard$},inner sep=0pt] {};
\end{axis}
\end{tikzpicture}
\caption{Polyhedral faces from an $\mathbb{R}^2$ Jordan block subspace, where 
  $(\lambda_1^n, \lambda_2^n)$ so that $\lambda_1=0.8, \lambda_2=0.8, 1 \leq n \leq 15$.
The red dot specifies an abstract vertex violating the safety specification (dashed blue line).
The closest iteration to the violating vertex is n=3.
A new support function (green) based on n=3 eliminates the violating vertex.
The new abstract polyhedra meets the safety specification (yellow).}
\label{fig:cegar_jordan_supports}
\end{figure}
}
\jrronly{
\subsection{Using Bounded Concrete Runs}\label{sec:cegar_concrete}
Due to the complex interactions in the dynamics, the above procedure may not always find the correct iterations 
for refinement, or at least not optimal ones. For this reason, a second method is proposed that in most cases 
will be more efficient and precise when the dynamics are strictly convergent.
This second approach relies on the direct calculation of the initial $\overline{k}$ iterations. Since we operate
over the eigenvalues and we limit $\overline{k}$ to a conservative bound, this is a relatively inexpensive calculation.
The approach leverages on the idea that for convergent dynamics, counterexamples are often found in the initial runs.
The first step is to directly calculate the trajectory of the counterexample for the first $\overline{k}$ iterations, and its
corresponding support function in the direction of $\vec{v}$. Once again, because this is a single point and a bounded 
time, this operation is comparatively inexpensive. 
The second step consists of finding an upper bound for all subsequent iterations, which we can do by using the norms
of the eigenvalues and the peaks of each geometrical multiple of a Jordan Block (which relate to these norms).
By selecting the larger between these two supports, we ensure soundness over the infinite time horizon.
This is equivalent to evaluating the reach tube as $X_n^\sharp=\bigcup\limits_{k=0}^{\overline{k}} \mat{A}^kX_0 \cup \mathcal{A}_{n-\overline{k}}\mat{A}^{\overline{k}}X_0$.

Since the above result is know to be an upper bound for the support in the direction of $\vec{v}$ we can directly add it to the inequalities of $\mathcal{A}$. 

\rronly{
\subsection{Case Study}

We have taken an industrial benchmark `CAFF problem Instance: Building' from
the competition forum (\footnote{\url{http://cps-vo.org/node/30277}}).  The
benchmark consists of a continuous model with 48 state variables and one
input.  Furthermore, there is an initial state corresponding to a
10-dimensional hyperrectangle.  The model is discretised using a 5\,ms sample
time to give us a discrete time model for verification.  Notice that the
choice of sample time has very little effect on abstract acceleration.  It
mainly affects the requirement for floating point precision (as very small
angles may require higher precision), and may have an effect on
counterexample generation which can either decrease precision or increase
time-cost based on some algorithmic choices.  The provided model requires an
analysis on the $25^{th}$ variable, with a safety specification requiring it
to remain below $.005$.  The problem has been verified using
SpaceEx\footnote{\url{http://spaceex.imag.fr}} in under 60 seconds.

The tool was run on this benchmark using different parameters.  We used an
Intel 2.6\,GHz I7 processor with 8\,GB of RAM running on Linux.  Although
the algorithm lends itself to concurrency, the tool currently supports only
single threading.  The process itself uses 82\,MB on this particular
benchmark.  The results are summarised in Table~\ref{table:results}.  It is
worth noting that for precisions under 1024 bits the tool returns soundness
errors when using sound arithmetic.


\begin{table*}[t!]
\centering
\footnotesize
\begin{tabular}{|l|*{5}{@{\;\;}c@{\;}}|}
\hline
Parameters &Sound&Inputs&Bits&Time&Bound \\  \hline \hline
-mp 128 -params "p=48,q=1,l=2" -templates "$x25$" & No &P & 128 & $13s$ & $0.013633$\\  \hline
-mp 128 -params "p=48,v=1:1,l=2" -templates "$x25$" & No  & V & 128 & $24s$ & $0.013716$\\  \hline \hline
-mp 128 -params "p=48,q=1,l=2" -sguard "$x25<.005$" & No  & P & 128 & $29s$ & $0.004996$\\  \hline
-mp 128 -params "p=48,v=1:1,l=2" -sguard "$x25<.005$" & No  & V & 128 & $48s$ & $0.004976$\\  \hline
-mp 1024 -params "p=48,q=1,l=2" -sguard "$x25<.005$" & No  & P & 1024 & $66s$ & $0.004996$\\  \hline
-mp 1024 -params "p=48,v=1:1,l=2" -sguard "$x25<.005$" & No  & V & 1024 & $90s$ & $0.004999$\\  \hline \hline
-mpi 1024 -params "p=48,q=1,l=2" -sguard "$x25<.005$" & Yes  & P & 1024 & $190s$ & $0.004996$\\  \hline
-mpi 1024 -params "p=48,v=1:1,l=2" -sguard "$x25<.005$"& Yes  & V & 1024 & $435s$ & $0.004998$\\  \hline
\end{tabular}~\\[0.5ex]
\caption{Tool Performance on Building Benchmark. P=Parametric, V=Variable}
\label{table:results}
\end{table*}
}
}
\jrronly{
 \section{Over-approximating continuous time dynamics using discrete models}\label{sec:cont_aa}
Thus far we wave discussed the Abstract Acceleration of discrete time systems. However, in many occasions, 
these models derive from the discretisation of continuous time systems. We therefore seek to establish a sound
over-approximation which does not only encompass the selected discretisation, but any chosen discretisation in the 
time domain.

\subsection{Linear Time Invariant Dynamical Systems}\label{sec:continuous_LTI}

A dynamical system is a system in which a function describes the progression of a state over time. 
In a continuous domain with linear dynamics, it is described by a first order Ordinary Differential Equation (ODE).
\begin{equation}
\dot{x}(t)=\mat{A}_o\vec{x}(t)+\mat{B}_o\vec{u}(t)
\label{eq:dynamical}
\end{equation}

Furthermore, a control system may have a derived output that is a linear combination of its states and inputs, 
which may restrict the observability of the state space from the output space.
\begin{equation}
\vec{y}(t)=\mat{C}_o\vec{x}(t)+\mat{D}_o\vec{u}(t)
\end{equation}

Discretisation of a continuous dynamical system turns the ODE into a difference equation in the form
\begin{align}
\label{eq:discretization}
\vec{x}_{k+1} &= \mat{A}_d\vec{x}_k+\mat{B}_d\vec{u}_k\\
y_k &= \mat{C}_d \vec{x}_ k + \mat{D}_d \vec{u}_ k 
\end{align}
where
\begin{align}
\label{eq:discretize}
\mat{A}_d &= e^{\mat{A}_o T_s} = \mathcal{L}^{-1} { ( s \mat{I} - \mat{A}_o )^{-1} }_{t = T_s}\\
\mat{B}_d &= \int_{t = 0}^{T_s} e^{\mat{A}_o t} dt\ \mat{B}_o = \mat{A}_o^{-1} ( \mat{A}_d - \mat{I} ) \mat{B}_o\\
\mat{C}_d &= \mat{C}_o\\
\mat{D}_d &= \mat{D}_o
\end{align}
and $T_s$ is the sample time. Then
\begin{align*}
x(kT)=x_k \text{ and } y(kT) = y_k, \forall k
\end{align*}
Let $g_d(k)=\mathcal{G}(t,g(t),T)$ be a function that performs the discretisation described above, where $g(t)$
represents the continuous dynamics and $g_d(k)$ the corresponding discrete dynamics. 
Let $G(s)=\mathcal{L}(g(t))$ and $G_d(z)=\mathcal{Z}(g_d(k))$ be the corresponding Laplace and Z-transforms
of $g(t)$ and $g_d(k)$. Given this relation, we have 
$$G_d(z)=G(z)|_{z=e^{sT}} : g_d(k)=\mathcal{G}(t,g(t),T) \wedge T < \frac{1}{0.5f_s}$$
Where $0.5f_s$ is the Nyquist frequency of $g(t)$. This last restriction is introduced to avoid the effects of aliasing
which could cause 'phantom poles' to appear otherwise. 
The eigenvalues of $\vec{A}_d$ corresponding to the poles of $G_d(z)$, and those of $\vec{A}$ corresponding
to the poles of $G(s)$ are similarly related.
$\hat{\lambda}_i=e^{-\lambda_iT}$ where $\hat{\lambda}_i \in \sigma(\vec{A}_d), \lambda_i \in \sigma(\vec{A})$
and $\sigma(\cdot)$ is the spectrum of a matrix.
The following remark is worth mentioning regarding the above.
\begin{remark}
The witnessed maximum amplitude of the discrete signals $x_k,y_k$ may be smaller to that of $x(t),y(t)$ due to
synchronism at fraction-frequency sampling. This means that reasoning about the state space must consider the
effects of this possibility as well of the case of maximal inputs from the continuous specification.
\end{remark}
Since \eqref{eq:discretization} is a bisimulation of \eqref{eq:dynamical}, we may use the semantics in section
\ref{sec:model_semantics} to model continuous dynamical systems.

 \subsection{A discrete time model using reals}\label{sec:real_aa}

 \subsubsection{Discrete dynamics using reals for systems without inputs}\label{sec:real_discrete_no_inputs}
 \begin{theorem}
 Given $\dot{\vec{x}}=\mat{A}\vec{x}$, where $\mat{A}=\mat{S}\mat{J}\mat{S}^{-1}$, the expression
 \begin{align}
 \vec{x}_T&=\vec{x}(t=T)=\mat{A}_{T}\vec{x}_0\\
 \mat{A}_{T}&= \mat{S}
 \left [ \begin{array}{cccc}
 e^{T\lambda_1}  & s_1\frac{T^{1}e^{T\lambda_i}}{(2)!} & \hdots  & s_i\frac{T^{p-1}e^{T\lambda_i}}{(p-i)!} \\
0 & e^{T\lambda_i}  & s_i\frac{T^{j-i}e^{T\lambda_i}}{(j-i)!} & \vdots \\
\vdots & & \ddots & \vdots \\
0 & \cdots & 0  &e^{T\lambda_i} \\
\end{array} \right ]
 \mat{S}^{-1}
 \label{eq:continuous_tube_dyn}\\
 &\text{where } s_i=\left\{\begin{array}{cc}1&gm(\lambda_i)>1\\0&gm(\lambda_i)=1\end{array}\right.,\nonumber
 \end{align}
$\lambda_i \in \mat{J}$ are the eigenvalues of $\mat{A}$, and $gm(\lambda_i)$ is the geometric multiplicity of $\lambda_i$.  $\vec{x}_T$ is a witness of the system $\dot{\vec{x}}(t)=\mat{A}\vec{x}$ at time t=T.
 \end{theorem}
 \begin{corollary}
 If $\mat{J}$ is diagonal and there exists a discrete dynamics matrix for a sampling time $T_d :  A_d=e^{\mat{A} T_d}$, then $\vec{x}_k=A_d^k\vec{x} : k \in (0,\infty)$ is a witness of the system at time $t=kT_d$.
 \end{corollary}
 \begin{proof}
 Let us recall equation \ref{eq:discretize}. The discrete representation of the system dynamics discretised with sample time $T_1$ is ruled by the formula
 $\mat{A}_1 = e^{\mat{A} T_1}$. From matrix theory, we have 
\begin{align}
 e^{\mat{A}}&=\sum_{k=0}^\infty \frac{1}{k!}\mat{A}^k
\end{align} 
\begin{align} 
 e^{\mat{S}\mat{J}\mat{S}^{-1}}&=\sum_{k=0}^\infty \frac{1}{k!}\left(\mat{S}\mat{J}\mat{S}^{-1}\right)^k
 =\mat{S} \left (\sum_{k=0}^\infty \frac{1}{k!}\mat{J}^k\right) \mat{S}^{-1}\nonumber\\
 &=\mat{S}e^{\mat{J}}\mat{S}^{-1}
\end{align} 
\begin{equation}
\mat{A}=\mat{S}\mat{J}\mat{S}^{-1} \Rightarrow \vec{A}_1 = \mat{S}\mat{J}_1\mat{S}^{-1}= \mat{S}e^{\mat{J} T_1}\mat{S}^{-1}.
 \end{equation}
 
 Let $\mat{J}$ be a diagonal matrix, such that 
 $${\lambda_1}_i=e^{\lambda_i T_1}, \forall \lambda_i \in \mat{J}.$$
 Let us now take a different sample rate $T_2$, such that 
 $${\lambda_2}_i=e^{\lambda_i T_2}, \forall \lambda_i \in \mat{J}.$$
 We can then say that 
 \begin{equation}
 {\lambda_2}_i=e^{\lambda_i T_1 \frac{T_2}{T_1}}={\lambda_1}_i^{\frac{T_2}{T_1}} \Rightarrow A_2=A_1^{\frac{T_2}{T_1}}.
 \end{equation}
 This proves the theorem for $gm(\lambda_i)=1$.
 We now note that given an exponentiated Jordan form with geometric multiplicity, the upper diagonal terms can be
 referenced to the original eigenvalues modified by the sampling rate.
  Let $\mat{J}\in \mathbb{R}^{p \times p}$ be a Jordan Canonical matrix whose powers are described in Equation \eqref{jord:pow}
 Then
\begin{align}
 \mat{J}_1&=\sum_{k=0}^\infty \frac{1}{k!}\left(\mat{J}T_1\right)^k\nonumber\\
 &=\sum_{k=0}^\infty \frac{1}{k!} \left [ \begin{array}{cccc}
 \lambda_i^k  & \binom{k}{1}  \lambda^{k-1} & \hdots  & \binom{k}{p-1} \lambda_i^{k-p+1} \\
& \lambda_i^k  & \binom{k}{1}  \lambda_i^{k-1} & \vdots \\
\vdots & \ddots & \ddots & \vdots \\
& &  &\lambda_i^k \\
\end{array} \right ] T_1^k
\end{align}
The eigenvalues remain the same as in the diagonal case, but the upper triangular terms are of the form:
\begin{align}
\forall j>i, c_{ij}&=\sum_{k=j-i}^\infty \frac{1}{k!}\binom{k}{j-i} \lambda_i^{k-(j-i)}T_1^k\\
&=\frac{1}{(j-i)!}\sum_{k=j-i}^\infty \frac{1}{(k-(j-i))!} \lambda_i^{k-(j-i)}T_1^k\nonumber\\
&=\frac{T_1^{j-i}e^{\lambda_i T_1}}{(j-i)!}\nonumber
\end{align}
which completes the proof for $gm(\lambda_i)>1$.
\end{proof}

 \subsubsection{Discrete dynamics using reals for systems with parametric inputs}\label{sec:real_discrete_param_inputs}
\begin{theorem}
The expression
 \begin{align}
 \vec{x}_T&=\vec{x}(t=T)=\mat{A}^{-1}\vec{x}_T'\nonumber\\
\vec{x}_T'&=\mat{A}\mat{A}_T\vec{x}_0 + (\mat{I}-\mat{A}_T)\mat{B}\vec{u} : \forall t \leq T,\ \vec{u}(t)=\vec{u} 
 \end{align}
 with $\mat{A}_T$ as per \eqref{eq:continuous_tube_dyn} is a witness of the system $\dot{\vec{x}}(t)=\mat{A}\vec{x}(t)+\mat{B}\vec{u}(t)$ at time $T$ given a parametric input $\vec{u}$.
 \end{theorem}
 \begin{corollary}
 If $\mat{J}$ is diagonal and there exists a discrete dynamics matrix for a sampling time $T_d :  A_d=e^{\mat{A} T_d}$, then $\vec{x}_k=A_d^k\vec{x}+(\mat{I}-\mat{A}_d^k\mat{B}\vec{u}) : k \in (0,\infty)$ is a witness of the system at time $kT_d$.
 \end{corollary}
 \begin{proof}
 We begin by expanding the equation
 \begin{align}
 \vec{x}_T&=\mat{A}^{-1}\vec{x}_T'\nonumber\\
 &=\mat{A}_T\vec{x}_0 + \mat{A}^{-1}(\mat{I}-\mat{A}_T)\mat{B}\vec{u}\nonumber\\
 &= \mat{A}_T\vec{x}_0 + \mat{B}_T\vec{u}
 \end{align}
 where $\mat{A}_T=e^{\mat{A}T}$ and $\mat{B}_T=\mat{A}^{-1}(\mat{I}-\mat{A}_T)\mat{B}$ correspond to $\mat{A}_d$ and $\mat{B}_d$ in equation \eqref{eq:discretize}
 \end{proof}
  
 \subsubsection{Discrete dynamics using reals for systems with variable inputs}\label{sec:real_discrete_var_inputs}

 \begin{theorem}
 Let $U_c=\mat{S}\vec{u}_c'$ be the centre of the eigenspace interval hull of $\mat{B}U$ as described in section \ref{sec:var_inpur_accel}, and $U_d=\mat{B}U-U_c$.
The expression
 \begin{align}
\label{eq:contu_reachset}
X_T^\sharp &=\mat{A}^{-1}X_T'\supseteq X_T \text{ where }\\
X_T'&=\mat{A}\mat{A}_TX_0 \oplus F_b^*\left((\mat{I}-\mat{A}_T),U_d\right) \oplus (\mat{I}-\mat{A}_T)U_c \nonumber
\end{align}
 is an over-approximation of the reach set  at time $t=T$ of the system $\dot{\vec{x}}(t)=\mat{A}\vec{x}(t) + \mat{B}\vec{u}(t)$ given $\vec{x}(0) \in X_0$ and $\vec{u}(t) \in U$.
 \end{theorem}
\begin{corollary}
 If $\mat{J}$ is diagonal and there exists a discrete dynamics matrix for a sampling time $T_d$ with  $A_d=e^{\mat{A} T_d}$, then $ \forall k \in (0,\infty)$:
 \begin{align}
 X_k\subseteq A_d^kX_0 &\oplus F_b^*\left((\mat{I}-\mat{A}_d^k)F_b^*\left((\mat{I}-\mat{A}_d)^{-1},U_d\right)\right)\nonumber\\
 &\oplus (\mat{I}-\mat{A}_d^k)(\mat{I}-\mat{A}_d)^{-1}U_c
 \end{align}
  is a witness of the system at time $kT_d$, where $\mat{A}_{bd}$ is a semi-spherical over-approximation of $\mat{A}_d$ as described in Equation \eqref{eq:discretize}.
 \end{corollary}
 \begin{proof}
Let $U_J'=\mat{S}^{-1}\mat{B}U$, with $\vec{u}_c'$ its geometrical centre and $U_b'$ the smallest semi-sphere centred at $\vec{u}_c'$ containing all points in $U_J'$ (see equation \eqref{eq:overball}).
Then $U_J'\subseteq U_b' \oplus \vec{u}_c' \Rightarrow (\mat{I}-\mat{A}_T)U' \subseteq (\mat{I}-\mat{A}_T)U_b' \oplus (\mat{I}-\mat{A}_T)\vec{u}_c'$. Since $ (\mat{I}-\mat{A}_{bT})U_b'\supseteq  (\mat{I}-\mat{A}_T)U_b'$, then $(\mat{I}-\mat{A}_T)U' \subseteq (\mat{I}-\mat{A}_{bT})U_b' \oplus (\mat{I}-\mat{A}_T)\vec{u}_c'$ thus proving the theorem for $\vec{u}(t)=\vec{u}_0 : t \in [0\ T]$.

For the continuous time interval, let us look at the solution for the differential equation at time $T$ with a non-deterministic input $\vec{u}$.
\begin{equation*}
\vec{x}_T=e^AT\vec{x}_0+\int_{t = 0}^{T_s} e^{\mat{A} t} \mat{B} \vec{u}(t) dt
\end{equation*}
We will first divide this equation into 
\begin{align*}
\vec{x}_T&=e^{AT}\vec{x}_0+\int_{t = 0}^{T} e^{\mat{A} t} \mat{B} \vec{u}_c dt +\int_{t = 0}^{T_s} e^{\mat{A} t} \mat{B} (\vec{u}(t)-\vec{u}_c) dt\\
&=A_T\vec{x}_0+B_T\vec{u}_c+\int_{t = 0}^{T} \sum_{k=0}^\infty \frac{1}{k!} \mat{A}^k t^k \mat{B}(\vec{u}(t)-\vec{u}_c) dt
\end{align*}
It follows from equation \eqref{eq:overball} that $\mat{A}^k \mat{B}(\vec{u}(t)-\vec{u}_c) \subseteq F_b^*(\mat{A}^k, \mat{B}(U-\vec{u}_c))$, hence
\begin{align*}
X_T \subseteq &A_TX_0+B_T\vec{u}_c\\
&+\int_{t = 0}^{T} \sum_{k=0}^\infty \frac{1}{k!}t^k F_b^*\left(\mat{A}^k, \mat{B}(U-\vec{u}_c)\right) dt\\
\subseteq& A_TX_0+B_T\vec{u}_c+\int_{t = 0}^{T_s} F_b^*\left(e^{\mat{A} t},\mat{B}(U-\vec{u}_c)\right) dt\\
\subseteq& A_TX_0+B_T\vec{u}_c+ \mat{A}^{-1}F_b^*\left((\mat{I}-\mat{A}_{T}),\mat{B}(U-\vec{u}_c)\right)
\end{align*}
Multiplying by $\mat{A}$ we get $X_T'$ as in \eqref{eq:contu_reachset}, thus proving the theorem.
\end{proof}

 \subsection{Finding Abstract Supports for Continuous Dynamics}\label{sec:cont_aasup}

In the case of discrete dynamics, support functions can be found by evaluating hyperplanes of consecutive exponents of the eigenvalues using differences.
In the case of continuous time, this translates to calculating derivatives. Once again, when the sets are known to be convex,
we may calculate the normal vector to this derivative and evaluate its support function on the selected iteration. The direction of the vector is determined by evaluating the support function of another point in the curve, which is known to be smaller than that of the selected point.
The corresponding derivatives are:
\begin{enumerate}
\item Positive Real Eigenvalues\\
We first calculate the slope of the tangent between the two progressions: $\frac{d f(k)}{d g(k)}=\frac{f'(k)}{g'(k)}$. Hence $\frac{d \lambda_1^k}{d \lambda_2^k}=\frac{log(\lambda_1) \lambda_1^k}{log(\lambda_2) \lambda_2^k}$. From this we know that the polar angle of the support vector is $\psi = \tan^{-1}(\frac{log(\lambda_1) \lambda_1^k}{log(\lambda_2) \lambda_2^k})+\frac{\pi}{2}$
\item Complex Conjugate Eigenvalue pairs\\
The pair forms a logarithmic spiral. The tangent of a spiral with polar equation $r=a e^{b\theta}$ will have a derivative $\frac{d r}{d \theta}=b a e^{b\theta}=b r$, hence the angle of the vector normal to the curve is $\tan^{-1}(\frac{r}{\frac{d r}{d \theta}})+\frac{\pi}{2}=\tan^{-1}(\frac{1}{b})+\frac{\pi}{2}$
\item Equal Eigenvalues\\
The supports are as on the discrete time case orthogonal to the plane $x=y$ and within the square containing the largest and lowest values for the pair.
\item Negative Real Eigenvalues and mixed types\\
Since we use the norm of the eigenvalues, we compute the tangents as in the case for positive real eigenvalues and apply the mirror image as in the discrete time case.
\end{enumerate}
}

 \subsection{Calculating the Number of Iterations for Continuous Dynamics}\label{sec:cont_guards}

Following the same steps as in Lemma \ref{lemma:jordan_guards} we develop an equivalent for continuous dynamics.
\begin{lemma}
\label{cont_jordan_guards}
Given a matrix $\mat{A}$ with eigenvalues $\{\lambda_s \mid s\in [1, r]\}$, where each eigenvalue $\lambda_s$ has a geometric multiplicity $p_s$ and corresponding generalised eigenvectors $\{\vec{v}_{s,i} \mid i \in [1, p_s]\}$,
\begin{align}
\forall t \geq 0, \mat{A}_t\vec{v}_{s}^{i}&=e^{\lambda_{s}}\vec{v}_{s,i}+\sum_{j=1}^{i-1} t^je^{\lambda_{s}}\vec{v}_{s,i-j}\nonumber\\
&=e^{\lambda_{s}}\left( \vec{v}_{s,i} + \sum_{j=1}^{i-1} t^j\vec{v}_{s,i-j}\right)
\label{eq:cont_jordan_iters}
\end{align}  
\end{lemma}
\begin{proof}
The proof derives again from the taylor expansion.
\begin{align}
e^{\mat{A}t}\vec{v}_{s,i}&=\sum_{k=0}^\infty \mat{A}^k\frac{t^k}{k!}\vec{v}_{s,i}=\sum_{k=0}^\infty \frac{t^k}{k!}\left( \lambda_s^k\vec{v}_{s,i}+k\lambda_s^{k-1}\vec{v}_{s,i-1}\right)\nonumber\\
&=\sum_{k=0}^\infty \frac{t^k}{k!}\lambda_s^k\vec{v}_{s,i}+\sum_{k=0}^\infty \frac{t^k}{k!}k\lambda_s^{k-1}\vec{v}_{s,i-1}\nonumber\\
&=e^{\lambda_it}(\vec{v}_{s,i}+t\vec{v}_{s,i-1})
\end{align}
The rest of the proof follows the same expansion as in \ref{lemma:jordan_guards} 
\end{proof}
Given the similarity of equation \eqref{eq:cont_jordan_iters} with \eqref{eq:jordan_iters} we may apply exactly the same techniques described in section \ref{sec:guards_geometric} to the continuous case.
\section{Experimental Results}

The algorithm has been implemented in C++ using the eigen-algebra package (v3.2), 
with double precision floating-point arithmetic, 
and has been tested on a 1.6\,GHz core~2 duo computer. 

\begin{table*}[t]
\centering
\footnotesize
\begin{tabular}{|l|*{4}{@{\;}c@{\;}}|*{2}{r@{\;\;}}|*{3}{r@{}l@{\;\;}}|*{2}{r@{}l@{\;\;\;}}|}
\hline
& \multicolumn{4}{c|}{characteristics}
& \multicolumn{2}{c|}{improved}
& \multicolumn{6}{c|}{analysis time [sec]} \\
name            & type                      & dim & inputs & bounds & IProc & Sti & \multicolumn{2}{c}{IProc}
& \multicolumn{2}{c}{Sti} & \multicolumn{2}{c|}{J+I}\\ \hline
parabola\_i1& $\neg s$,$\neg c$,$g$  & 2 & 1 & 80 & +25 & +28 & 0&.007 & 237& & 0&.049 \\
parabola\_i2& $\neg s$,$\neg c$,$g$  & 2 & 1 & 80 & +24 & +35 &  0&.008 & 289& & 0&.072 \\
cubic\_i1& $\neg s$,$\neg c$,$g$      & 3 & 1 & 120 &  +44 & +50 & 0&.015 & 704& & 0&.097 \\
cubic\_i2& $\neg s$,$\neg c$,$g$      & 3 & 1 & 120 & +35 & +55 & 0&.018 & 699& & 0&.124 \\
oscillator\_i0& $s$,$c$,$\neg g$        & 2 & 0 & 56 & +24 & +24 & 0&.004 & 0&.990 & 0&.021 \\
oscillator\_i1& $s$,$c$,$\neg g$        & 2 & 0 & 56 & +24 & +24 & 0&.004 & 1&.060 & 0&.024 \\
inv\_pendulum& $s$,$c$,$\neg g$    & 4 & 0 & 16 & +8 & +8 & 0&.009 & 0&.920 & 0&.012 \\
convoyCar2\_i0 & $s$,$c$,$\neg g$ & 3 & 2 & 12 &   +9 & +9 & 0&.007 & 0&.160 & 0&.043 \\
convoyCar3\_i0 & $s$,$c$,$\neg g$ & 6 & 2 & 24 & +15 & +15 & 0&.010 & 0&.235 & 0&.513 \\
convoyCar3\_i1 & $s$,$c$,$\neg g$ & 6 & 2 & 24 & +15 & +15 & 0&.024 & 0&.237 & 0&.901 \\
convoyCar3\_i2 & $s$,$c$,$\neg g$ & 6 & 2 & 24 & +15 & +15 & 0&.663 & 0&.271 & 1&.416 \\
convoyCar3\_i3 & $s$,$c$,$\neg g$ & 6 & 2 & 24 & +15 & +15 & 0&.122 & 0&.283 & 2&.103 \\ \hline
\end{tabular}\\[0.7ex]
\textbf{type}: \textbf{$s$} -- stable loop, \textbf{$c$} -- complex eigenvalues, \textbf{$g$} -- loops with guard;
\textbf{dim}: system dimension (variables); \textbf{bounds}: nb. of half-planes defining the polyhedral set; \\
\textbf{IProc} is \cite{jeannet2010interproc}; \textbf{Sti} is \cite{colon2003linear}; \textbf{J+I} is this work; \\
\textbf{improved}: number of bounds newly detected by J+I over the existing tools~(IProc,~Sti)
\caption{Experimental comparison of unbounded-time analysis tools with inputs}
\label{tab:exp1}
\end{table*}

\subsection{Comparison with other unbounded-time approaches.} 

\begin{table*}[t]
\centering
\footnotesize
\begin{tabular}{|l|lcc|rr|r@{}lr@{}l|@{\;\;}r@{}lr@{}l|@{\;\;}r@{}l|@{\;\;}r@{}l|}
\hline
& \multicolumn{3}{c|}{characteristics}
& \multicolumn{2}{c|}{improved}
& \multicolumn{12}{c|}{analysis time (sec)} \\
name            & type                      & dim & bounds & tighter& looser
 & \multicolumn{4}{c}{\quad J\quad (jcf)} &\multicolumn{4}{l}{mpfr+(jcf) }
& \multicolumn{2}{l}{mpfr} & \multicolumn{2}{c|}{ld}\\ \hline
parabola\_i1& $\neg s$,$\neg c$,$g$  & 3 & 80 &+4(5\%) & 0(0\%) & 2&.51 &( 2&.49) & 0&.16&(0&.06) & 0&.097 & 0&.007 \\
parabola\_i2& $\neg s$,$\neg c$,$g$  & 3 & 80 &+4(5\%) & 0(0\%) & 2&.51 &( 2&.49) & 0&.26&(0&.06)& 0&.101 & 0&.008 \\
cubic\_i1& $\neg s$,$\neg c$,$g$        & 4 & 120 & 0(0\%) & 0(0\%) & 2&.47 &( 2&.39) & 0&.27 &(0&.20) & 0&.110 & 0&.013 \\
cubic\_i2& $\neg s$,$\neg c$,$g$        & 4 & 120 & 0(0\%) & 0(0\%) & 2&.49 &( 2&.39) & 0&.32&(0&.20) & 0&.124 & 0&.014 \\
oscillator\_i0& $s$,$c$,$\neg g$          & 2 &  56 & 0(0\%) & -1(2\%) & 2&.53 &( 2&.52) & 0&.12 &(0&.06)& 0&.063 & 0&.007 \\
oscillator\_i1& $s$,$c$,$\neg g$          & 2 &  56 & 0(0\%) & -1(2\%) & 2&.53 &( 2&.52) & 0 &.12 &(0&.06)& 0&.078 & 0&.008 \\
inv\_pendulum& $s$,$c$,$\neg g$ & 4 & 12 & +8(50\%) & 0(0\%) & 65&.78 &(65&.24) & 0&.24 & (0&.13) & 0&.103 & 0&.012 \\
convoyCar2\_i0 & $s$,$c$,$\neg g$ & 5 & 12 & +9(45\%) & 0(0\%) & 5&.46 &( 4&.69) & 3&.58&(0&.22) & 0&.258 & 0&.005 \\
convoyCar3\_i0 & $s$,$c$,$\neg g$ & 8 & 24 & +10(31\%) & -2(6\%) & 24&.62 & (11&.98) & 3&.11&(1&.01) & 0&.552 & 0&.051 \\
convoyCar3\_i1 & $s$,$c$,$\neg g$ & 8 & 24 & +10(31\%) & -2(6\%)& 23&.92 & (11&.98) & 4&.94&(1&.01) & 0&.890 & 0&.121 \\
convoyCar3\_i2 & $s$,$c$,$\neg g$ & 8 & 24 & +10(31\%) & -2(6\%) & 1717&.00 & (11&.98) & 6&.81&(1&.01) & 1&.190 & 0&.234 \\
convoyCar3\_i3 & $s$,$c$,$\neg g$ & 8 & 24 & +10(31\%) & -2(6\%)  & 1569&.00 & (11&.98) & 8&.67&(1&.01) & 1&.520 & 0&.377 \\ \hline
\end{tabular}\\[0.7ex]
{\footnotesize\textbf{type}: \textbf{$s$} -- stable loop, \textbf{$c$} -- complex eigenvalues, \textbf{$g$} -- loops with guard;
\textbf{dim}: system dimension (including fixed inputs); \textbf{bounds}: nb. of half-planes defining the polyhedral set; 
\textbf{improved}: number of bounds (and percentage) that were tighter (better) or looser (worse) than \cite{JSS14}; \\
\textbf{J} is \cite{JSS14}; \textbf{mpfr+} is this article using 1024bit
 mantissas ($e<10^{-152}$);  \textbf{mpfr} uses a 256bit mantissa
 ($e<10^{-44}$); \textbf{ld} uses a 64bit mantissa ($e<10^{-11}$); here $e$ is
 the accumulated error of the dynamical system; \textbf{jcf}: time
 taken to compute Jordan form \\[0.5ex]}
\caption{Experimental comparison with previous work}
\label{tab:exp2}
\end{table*}

In a first experiment we have benchmarked our implementation against the
tools~\textsc{InterProc}~\cite{jeannet2010interproc} and
\textsc{Sting}~\cite{colon2003linear}.  
We have tested these tools on
different scenarios, including guarded/unguarded, stable/un\-stable and
complex/real loops with inputs (details in Table~\ref{tab:exp1}).%
\footnote{The tool and the benchmarks are available from
\url{http://www.cprover.org/LTI/}.}
It is important to note that in many instances,  \textsc{InterProc}
\rronly{(due to the limitations of widening)} and \textsc{Sting}
\rronly{(due to the inexistence of tight polyhedral, inductive invariants)}
are unable to infer finite bounds at all.

Table~\ref{tab:exp2} gives the comparison 
of our implementation using different levels of precision (long
double, 256 bit, and 1024 bit floating-point precision) with
the original abstract acceleration for linear
loops without inputs (J)~\cite{JSS14} (where inputs are fixed to
constants). 
This shows that our implementation gives tighter over-approximations on most
benchmarks (column `improved').  While on a limited number of instances the
current implementation is less precise (Fig.~\ref{aa:supports} gives a hint
why this is happening), the overall increased precision is owed to lifting
the limitation on directions caused by the use of logahedral abstractions.

At the same time, our implementation is faster -- even when used with 1024
bit floating-point precision -- than the original abstract acceleration
(using rationals).  The fact that many bounds have improved with the new
approach, while speed has increased by several orders of magnitude, provides
evidence of the advantages of the new approach.

\rronly{
The speed-up is due to the faster Jordan form computation, which takes
between 2 and 65 seconds for~\cite{JSS14} (using the ATLAS package), whereas
our implementation requires at most one second.
For the last two benchmarks, the polyhedral computations blow up
in~\cite{JSS14}, whereas our support function approach shows only moderately
increasing runtimes.  The increase of speed is owed to multiple factors, as
detailed in Table~\ref{tab:speed}.  The difference of using long double
precision floating-point vs.~arbitrary precision arithmetic is negligible,
as all results in the given examples match exactly to 9 decimal places. 
Note that, as explained above, soundness can be ensured by appropriate
rounding in the floating-point computations.

\begin{table}[b!]
\centering
\footnotesize
\begin{tabular}{|l|*{1}{@{\;\;}c@{\;}}|}
\hline
Optimisation & Speed-up \\  \hline \hline
Eigen vs.~ATLAS \footnote{http://eigen.tuxfamily.org/index.php?title=Benchmark} & $2$--$10$ \\  \hline
Support functions vs.~generators
  & $2$--$40$ \\ \hline
long double vs.~multiple precision arithmetic & $5$--$200$ \\ \hline
interval vs.~regular arithmetic & $.2$--$.5$ \\ \hline \hline
Total & $4$--$80000$ \\ \hline
\end{tabular}~\\[0.5ex]
\caption{Performance improvements by feature}
\label{tab:speed}
\end{table}
 }
\subsection{Comparison with bounded-time approaches.}
In a third experiment, we compare our method with the LGG algorithm~\cite{LG09}
used by \textsc{SpaceEx}~\cite{FLD+11}.  In order to set up a fair
comparison we have provided the implementation of the native algorithm
in~\cite{LG09}.  We have run both methods on the convoyCar example~\cite{JSS14} with inputs, which presents an unguarded,
scalable, stable loop with complex dynamics, and focused on octahedral
abstractions.
For convex reach sets, the approximations computed by abstract acceleration
are quite tight in comparison to those computed by the LGG algorithm.
However, storing finite disjunctions of convex polyhedra, the LGG algorithm
is able to generate non-convex reach tubes, which are arguably more proper
in case of oscillating or spiralling dynamics.
Still, in many applications abstract acceleration can provide a tight
over-approximation of the convex hull of those non-convex reach sets.

Table~\ref{tab:spaceX} gives the results of this comparison. For simplicity,
we present only the projection of the bounds along the variables of interest.  
As expected, the LGG algorithm performs better in terms of tightness, 
but its runtime increases with the number of iterations.
Our implementation of LGG using Convex Polyhedra with octagonal
templates is slower than the abstractly accelerated version even for
small time horizons (our implementation of LGG requires $\sim\!4$\,ms
for each iteration on a 6-dimensional problem with octagonal
abstraction).  This can be improved by the use of zonotopes, or by
careful selection of the directions along the eigenvectors, but this
comes at a cost on precision. Even when finding combinations that
outperform our approach, this will only allow the time horizon of the
LGG approach to be slightly extended before matching the analysis time
from abstract acceleration, and the reachable states will still
remain unknown beyond the extended time horizon.

The evident advantage of abstract acceleration is its speed over finite
horizons without much precision loss, and of course the ability to prove
properties for unbounded-time horizons.

\begin{table*}[t]
\centering
\footnotesize
\begin{tabular}{|l|*{2}{@{\;\;\;}c@{\;}}|*{3}{c@{\;\;\;}}|}
\hline
& \multicolumn{2}{c|}{this article} & \multicolumn{3}{c|}{LGG}\\
name            & 100 iterations                     &\;\;\; unbounded \;\;\;& 100 iterations& 200 iterations & 300 iterations \\ \hline
run time &  166\,ms & 166\,ms & 50\,ms & 140\,ms & 195\,ms\\
car acceleration & [-0.820 1.31] & [-1.262 1.31] &[-0.815 1.31] & [-0.968  1.31] & [-0.968  1.31] \\
car speed & [-1.013 5.11] & [-4.515 6.15] &[-1.013 4.97] & [-3.651  4.97] & [-3.677  4.97] \\
car position & [43.7  83.4] & [40.86  91.9] & [44.5  83.4] & [44.5  88.87] & [44.5  88.87]\\ \hline
\end{tabular}~\\[0.5ex]
\caption{Comparison on convoyCar2 benchmark, between this work and
the LGG algorithm~\cite{LG09}}
\label{tab:spaceX}
\end{table*}

\jrronly{
\subsection{Comparison with other Abstract Acceleration techniques.}
Table \ref{tab:leguernic} shows a comparison between our approach and \cite{leguernic:hal-01550767}. 
As can be seen by the results, our approach is not only faster but much more precise. The reasons for
this are many-fold. In terms of speed the fact that they require a matrix twice the size and that the 
algorithm is O($n^3$) alredy makes a difference of an order of magnitude. Even when using sparse matrices
their larger number of coefficients will always result in a slower approach.
In terms of precision, the choice to separate the center of the inputs makes a considerable difference, 
which is increased by the fact that our circular overapproximations are contained within their interval hulls
and are therefore up to $\frac{4}{3}^n$ times smaller in volume. Since these are accelerated the encompassing increase 
in the size of the will be considearble.
Finally, we note that in the case below, where all original eigenvalues are convergent, the interval hull 
approach creates one divergent eigenvalue which causes a critical change in the behaviour of the dynamics
that leads to unbounded results.

\begin{table*}[t]
\centering
\footnotesize
\begin{tabular}{|l|*{2}{@{\;\;\;}c@{\;}}|*{2}{c@{\;\;\;}}|}
\hline
& \multicolumn{2}{c|}{this article} & \multicolumn{2}{c|}{Interval Hulls}\\
name            & 100 iterations                     &\;\;\; unbounded \;\;\;& 100 iterations& unbounded \\ \hline
run time &  166\,ms & 166\,ms & 155085\,ms & 155085\,ms \\
car acceleration & [-0.820 1.31] & [-1.262 1.31] &[-24.24 23.9] & [-$\infty$  $\infty$] \\
car speed & [-1.013 5.11] & [-4.515 6.15] &[-97.2 86.7] & [-$\infty$  $\infty$] \\
car position & [43.7  83.4] & [40.86  91.9] & [-319  343.4] & [-$\infty$  $\infty$]\\ \hline
\end{tabular}~\\[0.5ex]
\caption{Comparison on convoyCar2 benchmark, between this work and
the work in~\cite{leguernic:hal-01550767}}
\label{tab:leguernic}
\end{table*}
}

\jrronly{
\subsection{Comparison with continuous time approaches}
}

\subsection{Scalability}
%
Finally, in terms of scalability, we have an expected
$\mathcal{O}(n^3)$ complexity worst-case bound (from the matrix multiplications in equation \eqref{equ:absaccinput}).  
We have parameterised the
number of cars in the convoyCar example~\cite{JSS14} (also seen in
Table~\ref{tab:exp2}), and experimented with up to 33 cars (each car after the
first requires 3 variables, so that for example $(33-1)\times 3= 96$ variables), and 
have adjusted the initial states/inputs sets. We report an average of 10 runs
for each configuration. These results demonstrate that our method scales to industrial-size problems. 
\begin{center}
\small
\begin{tabular}{|l|*{6}{@{\;\;}c@{\;}}|}
\hline
\# of variables &  3 & 6 & 12 & 24 & 48 & 96\\
runtime (s) & 0.004 & 0.031 & 0.062 & 0.477 & 5.4 & 56 \\ \hline
\end{tabular}~\\[-0.5ex]
\end{center}
%

\section{Related Work} 

There are several approaches that solve the safety problem for the linear and other cases such as hybrid systems. 
They are broadly divided into two categories due to the inherent nature of these. Namely the time bounded analysis is in most cases
unsound since it cannot reason about the unbounded time case (we not that a proof of the existence of a fix-point for the given horizon would restore such soundness by many tools do not attempt to find such proof which is left to the user). Unbounded-time solutions are therefore preferred when such soundness is required, although they are often either less precise or slower than their bounded counterparts.
\subsection{Time-Bounded Reachability Analysis}\label{sec:bounded_RW}
The first approach is to surrender exhaustive analysis over the infinite
time horizon, and to restrict the exploration to system dynamics up to some
given finite time bound.  Bounded-time reachability is decidable, and
decision procedures for the resulting satisfiability problem have made much
progress in the past decade.  The precision related to the bounded analysis
is offset by the price of uncertainty: behaviours beyond the given time
bound are not considered, and may thus violate a safety requirement.
Representatives are STRONG~\cite{DRJ13}, HySon~\cite{bouissou2012hyson}, 
CORA~\cite{althoff2015introduction}, HYLAA~\cite{DBLP:conf/hybrid/BakD17}
and SpaceEx~\cite{FLD+11}.

Set-based simulation methods generalise guaranteed integration~\cite{Loe88,Bou08}
 from enclosing intervals to relational domains.  They use precise
abstractions with low computational cost to over-approximate sets
of reachable states up to a given time horizon.
Early tools used polyhedral sets (\textsc{HyTech}~\cite{HHW97} and
\textsc{PHAVer}~\cite{Fre05}), 
polyhedral flow-pipes~\cite{CK98}, 
ellipsoids~\cite{BT00} and zonotopes~\cite{Gir05}. 
A breakthrough was been achieved by~\cite{GLM06,LG09}, with the
representation of convex sets using template polyhedra and support
functions.  This method is implemented in the
tool~\textsc{SpaceEx}~\cite{FLD+11}, which can handle dynamical systems with
hundreds of variables. Although it may use exact arithmetic to maintain soundness,
it performs computations using floating-point numbers: this is a deliberate choice to
boost performance, which, although quite reasonable, its implementation is numerically
unsound and therefore does not provide genuine formal guarantees.
In fact, most tools using eigendecomposition over a large number of variables (more than 10)
are numerically unsound due to the use of unchecked floating-point arithmetic.
Another breakthrough in performance was done by HYLAA~\cite{DBLP:conf/hybrid/BakD17} which
was the first tool to solve all high order problems of hundreds and thousands dimensions.
Other approaches use specialised constraint solvers (HySAT~\cite{FH07},
iSAT~\cite{EFH08}), or SMT encodings~\cite{CMT12,GT08} for bounded model
checking of hybrid automata.

\subsection{Unbounded Reachability Analysis}\label{sec:unbounded_RW}

The second approach, epitomised in static analysis methods~\cite{HRP94},
explores unbounded-time horizons.  It employs conservative
over-approximations to achieve completeness and decidability over infinite
time horizons.

Unbounded techniques attempt to infer a \emph{loop invariant}, i.e., an
inductive set of states that includes all reachable states.  If the computed
invariant is disjoint from the set of bad states, this proves that the
latter are unreachable and hence that the loop is safe.  However, analysers
frequently struggle to obtain an invariant that is precise enough with
acceptable computational cost.  The problem is evidently exacerbated by 
non-determinism in the loop, which corresponds to the case of
open systems.  Prominent representatives of this analysis approach include
Passel~\cite{johnsonpassel}, Sting~\cite{colon2003linear}, and abstract
interpreters such as \textsc{Astr\'ee} \cite{BCC+03} and
InterProc~\cite{jeannet2010interproc}.
Early work in this area has used implementations of abstract interpretation
and widening~\cite{CC77}, which are still the foundations of most modern tools.
The work in~\cite{HRP94} uses 
abstract interpretation with convex polyhedra over piecewise-constant
differential inclusions.
Dang and Gawlitza~\cite{DG11a} employ optimi\-sation-based (max-strategy
iteration) with linear templates for hybrid systems with linear dynamics.
%
Relational abstractions \cite{ST11} use ad-hoc ``loop summarisation''
of flow relations, while abstract acceleration focuses on linear relations
analysis~\cite{GH06,GS14}, which is common in program analysis.  

\subsection{Abstract Acceleration}\label{sec:AA_RW}

Abstract acceleration~\cite{GH06,JSS14,GS14} captures the effect of an arbitrary
number of loop iterations with a single, non-iterative transfer function
that is applied to the entry state of the loop (i.e., to the set of initial
conditions of the linear dynamics).
Abstract acceleration has been extended from its original version to encompass inputs over
reactive systems~\cite{SJ12} but restricted to subclasses of linear
loops, and later to general linear loops but without inputs~\cite{JSS14}. 

The work presented in this article lifts these limitations by presenting 
abstract acceleration for \emph{general} linear loops \emph{with} inputs \cite{cattaruzza2015unbounded},
developing numeric techniques for scalability and extending the domain to continuous time systems.

\jrronly{
The work in \cite{leguernic:hal-01550767} claims the unsoundness of \cite{cattaruzza2015unbounded}. 
The sources of unsoundness referred to in that paper had been address before its publication in
\cite{extended-version} for algorithmic soundness and \cite{cattaruzza2017sound} regarding numerical soundness.
The paper proposes an alternative approach to Abstract Acceleration with inputs which relies on the expansion of
the dynamical equation to a 4x-dimensional model. The model presented there is slower and more imprecise than
the one presented in \cite{cattaruzza2015unbounded}. However, the handling of the guards with respect to calculating
the exact number of iterations would prove to be more precise in most cases. 
}
\section{Conclusions and Future Work} \label{sec:concl}

We have presented an extension of the Abstract Acceleration paradigm to
guarded LTI systems (linear loops) with inputs, overcoming the limitations
of existing work dealing with closed systems.  We have decisively
shown the new approach to over-compete state-of-the-art tools for
unbounded-time reachability analysis in both precision and scalability.  The
new approach is capable of handling general unbounded-time safety analysis
for large scale open systems with reasonable precision and fast computation
times. Conditionals inside loops and nested loops are out of the scope of
this paper.

Work to be done is extending the approach to non-linear
dynamics, which we believe can be explored via hybridisation
techniques~\cite{ADG07}, and to formalise the framework for general hybrid
models with multiple guards and location-dependent dynamics, with the aim to
accelerate transitions across guards rather than integrate individual
accelerations on either side of the guards.

\paragraph{\textsc{Acknowledgments}.} 
We would like to thank Colas Le Guernic  for his
constructive suggestions and comments on the paper.

\newpage
\bibliographystyle{abbrv} 
\bibliography{mybibliography}

\end{document}